\documentclass[letterpaper,12pt]{article}
\usepackage[margin=1in]{geometry}
\usepackage[T1]{fontenc}
\usepackage{lmodern}
\usepackage[onehalfspacing]{setspace}
\usepackage[authoryear]{natbib}
\usepackage{xcolor}
\usepackage[section]{placeins}
\usepackage{amsfonts,amssymb,bm,latexsym,amsmath,url,amsbsy,graphicx,booktabs,mathtools}
\newtheorem{theorem}{Theorem}[section]
\newtheorem{assumption}{Assumption}[section]
\newtheorem{proposition}{Proposition}[section]
\newtheorem{corollary}{Corollary}[section]
\newtheorem{lemma}{Lemma}[section]
\newtheorem{example}{Example}[section]

\newcommand{\E}{\mathbb{E}}
\renewcommand{\P}{\mathbb{P}}
\newcommand{\R}{\mathbb{R}}
\newcommand{\1}{\mathbf{1}}
\newcommand{\Pib}{\boldsymbol{\Pi}}        % menu of saturations {pi_1,...,pi_K}
\newcommand{\Sb}{\mathbf{S}}               % first-stage assignment vector
\newcommand{\Zb}{\mathbf{Z}}               % within-cluster treatment vector
\newcommand{\Uh}{\widehat{U}}              % estimated welfare
\newcommand{\Yb}{\bar{Y}}                  % cluster mean potential outcome
\newcommand{\Ybh}{\widehat{\bar{Y}}}       % observed cluster mean
\newcommand{\dES}{\delta_{\mathrm{ES}}}    % ES ranking rule
\newcommand{\Reg}{\mathcal{R}}             % risk (expected loss)
\newcommand{\Loss}{L}                      % welfare-weighted Kendall loss
\newcommand{\Normal}{\mathcal{N}}
\newcommand{\astar}{a^{*}}                 % 0.16997
\newcommand{\bstar}{b^{*}}                 % 0.11980
\newcommand{\MES}{M_{\mathrm{ES}}}             % maximum Gaussian risk of the ES ranking
\newcommand{\Ueff}{\widehat{U}^{\mathrm{eff}}} % efficient welfare estimator
\newcommand{\deff}{\delta^{\mathrm{eff}}}      % ranking of the efficient estimates
\newcommand{\Hnb}{\mathcal{B}}                 % Hellinger neighbourhood
\newcommand{\rev}[1]{#1}
\newcommand{\chg}{}

\newenvironment{proof}[1][Proof]{\par\addvspace{6pt}\noindent\textsc{#1.}\ \ignorespaces}{\par\addvspace{6pt}}
\usepackage[colorlinks=true,linkcolor=blue!60!black,citecolor=blue!60!black,urlcolor=blue!60!black]{hyperref}

\renewcommand{\thesection}{\arabic{section}}

\counterwithin*{equation}{section}

\title{Ranking Treatment Saturations:\\ Regret Bounds and Experimental Design}
\author{Seungjin Han\thanks{Department of Economics, McMaster University, Hamilton, ON, Canada. Email:
\href{mailto:hansj@mcmaster.ca}{hansj@mcmaster.ca}.}
\and Julius Owusu\thanks{Department of Economics, Concordia University, Montreal, QC, Canada. Email:
\href{mailto:julius.owusu@concordia.ca}{julius.owusu@concordia.ca}.}
\and Youngki Shin\thanks{Department of Economics, McMaster University, Hamilton, ON, Canada. Email:
\href{mailto:shiny11@mcmaster.ca}{shiny11@mcmaster.ca}.}}
\date{\today}

\begin{document}

\maketitle

\begin{abstract}
\chg \rev{We study the statistical decision problem of ranking treatment saturations, the treated shares of a cluster.
Planners need a ranking, not only the best saturation, when rollouts are phased or the preferred saturation becomes
infeasible. Using a two-stage randomized saturation experiment, we analyse the empirical success (ES) ranking, a simple
and widely used rule. We derive finite-sample regret bounds for ES under the welfare-weighted Kendall distance loss that
allow arbitrary within-cluster interference and need no network data. In a local asymptotic framework, ES on efficient
welfare
estimates has maximum regret at most about 1.4 times the minimax regret, for any number of saturations and any
covariance. This bound guides the allocation of clusters by balancing cost against the variance of cluster means under
each saturation. In the Malawi cash-transfer experiment, ES ranks full saturation of unconditional transfers first and
zero last for marriage, and the realized allocation nearly minimizes the bound.}


\end{abstract}

\clearpage
\section{Introduction}\label{sec:intro}

\chg Many programmes reach people through groups such as villages, schools or local labour markets, and their effects
depend
on how many members of a group are treated. \rev{We call such a group a cluster and its members, such as individuals or
households, units. The treated share of a cluster is its saturation, and saturation experiments vary this share across
clusters.}
\citet{crepon2013labor} randomized the share of eligible young job seekers offered job-search assistance across French
local labour markets, from 0 to 100 percent, to measure displacement effects. In Zomba, Malawi, \citet{baird2011cash}
offered cash transfers to nominal shares 0, 1/3, 2/3 and 1 of the baseline schoolgirls sampled in each enumeration
area. Such experiments are usually analysed for direct and spillover effects. A planner who must implement one
saturation \rev{from a finite menu of candidates} asks a different question: which saturation to use, and \rev{how the
candidates rank}.

\chg \rev{The planner often needs a ranking rather than only the best saturation.} Rollouts are phased under budgets
that relax
slowly, and legal, capacity or funding shocks can rule out the preferred saturation after the decision\rev{. The planner
then needs} the next-best option and the one after it. The planner must also justify the chosen saturation against the
alternatives, and the gaps between
adjacent saturations show whether the decision is sharp or more data are worth collecting. A ranking contains the
choice of the best saturation. The converse fails.

Three features make the problem harder than standard treatment choice. A unit's outcome can depend on the treatment
of everyone in its cluster, through peers, prices or marriage markets, so units are not the independent observations
that the treatment choice literature following \citet{manski2004statistical} assumes. The network that carries the
spillovers is rarely observed, so this dependence is hard to model directly. Finally, the welfare of a saturation is
defined under a hypothetical policy that treats a given share of every cluster, not by the potential outcome of any
single unit.

\chg We study the problem in a statistical decision framework. The welfare of a saturation is the expected cluster mean
outcome when the planner treats that share of every cluster at random. The empirical success (ES) ranking orders the
saturations by the average of the cluster means in each treatment arm. The loss is the welfare-weighted Kendall distance
(\citealp{kendall1938new}) between a ranking and the true order: each misordered pair costs its welfare gap. This loss
bounds the regret of choosing the top-ranked saturation \rev{for every realization of the data. Every bound on ranking
regret below is therefore} also a bound on \rev{treatment choice} regret.

\chg Our first result is a finite-sample regret bound that needs no knowledge of the network. Under i.i.d.\ sampling of
clusters and complete randomization at both stages of the design, the cluster means within each arm are independent
across clusters and arms, whatever the interference inside clusters and whatever the cluster sizes. Hoeffding's
inequality then bounds the probability that ES reverses a pair, and the ranking regret, for every number of clusters
$C$. The inequality is standard. What is new is that the independent units are cluster means, so the bound survives
arbitrary within-cluster interference. With fixed shares of clusters across saturations, the bound decays
exponentially in $C$ \rev{for any fixed welfare vector. The maximum regret of ES over outcome distributions is of order
$C^{-1/2}$, and this order is exact.} The bound gives a
simple sample-size rule: with $m$ clusters in each of two saturations whose welfare differs by $\Delta$,
$m\ge\log(1/\varepsilon)/\Delta^2$ clusters per saturation suffice to order the pair correctly with probability at least
$1-\varepsilon$.

\chg Our main result asks whether any rule can do substantially better. In a local asymptotic framework in which welfare
gaps
shrink at rate $C^{-1/2}$, as in \citet{hirano2009asymptotics}, the maximum regret of the ranking of efficient welfare
estimates is at most $1.42$ times the local asymptotic minimax regret, that is, at most 42\% above the smallest
worst-case regret that any rule can achieve. The factor holds for every number of saturations and every covariance of
the estimates, so ES, the rule practitioners already use, is close to optimal under interference. The upper bound sums
the two-action minimax value of \citet{hirano2009asymptotics} over pairs. The new part is a lower bound of the same
order. It comes from a Gaussian prior on the welfare vector whose covariance is proportional to that of the estimates.
This prior respects the links among the gaps (the first-to-third gap is the sum of the two adjacent ones), and under it
ES is the Bayes rule. To our knowledge this is the first explicit two-sided bound for the regret of ranking $K$
\rev{treatments} with an arbitrary covariance.

Raw arm means ignore cluster size and the assignment pattern, which are also observed, so the result concerns efficient
welfare estimates. The ES ranking of the raw arm means (raw ES) attains the factor when only cluster means are observed
or when neither feature predicts the cluster mean. With two saturations or an exchangeable covariance, ES is exactly
minimax in the limit experiment.

\chg The upper bound of Theorem~\ref{thm:nearminimax} also \rev{provides a criterion for allocating clusters across
saturations}. It is proportional to the sum over pairs of the standard deviations of the estimated welfare contrasts,
and minimizing it under a budget with per-cluster costs is a convex problem with a unique solution. Equal costs and
variances give the balanced design. With two saturations the counts are proportional to $(v_k/c_k)^{1/2}$, where $v_k$
is the variance of a cluster mean and $c_k$ the cost of a cluster, a Neyman-type allocation. \rev{This allocation
minimizes the bound rather than the regret itself. In the limit experiment that keeps only the arm and the mean outcome
of each cluster, ES under this allocation has maximum regret less than 1.42 times the smallest maximum regret that any
budget-feasible allocation and ranking rule can attain together. The allocation need not be exactly minimax.}

\chg In Zomba, 88 programme areas \rev{were assigned to one of four saturations of cash-transfer offers to schoolgirls,
with
transfers either conditional on school attendance (CCTs) or unconditional (UCTs).} For UCTs and \rev{the never-married
outcome}, ES puts full saturation first and zero last, the only gap in either menu
above 2 standard errors, while the intermediate saturations are not separated. The plug-in \rev{bound on the} worst-case
\rev{regret of ES at this design} is 0.06 to 0.07, and reallocating the areas would lower \rev{it} by only 2.3 to
3.6\rev{\%}.
Distribution-free guarantees at this precision would need thousands of clusters.

Our work belongs to the literature on statistical treatment choice. \citet{manski2004statistical} introduced
finite-sample regret analysis of the ES rule, \citet{stoye2009minimax} and \citet{tetenov2012statistical} derive minimax
regret rules, \citet{kitagawa2018should} and \citet{athey2021policy} study empirical welfare maximization, and
\citet{hirano2009asymptotics} develop the local asymptotic theory that we use. These papers assume no interference, and
they choose a treatment rather than rank. \citet{masten2023minimax} shows that ES is minimax for choosing one of several
treatments, in the limit experiment and locally asymptotically, when samples are balanced and variances are equal. He
leaves unequal precision open. \citet{chen2025note} study minimax regret rules with multiple treatments in finite
samples. Relative to these papers, we rank the whole menu, our factor of about 1.4 holds with unequal precision and
correlated estimates, and we derive the allocation of clusters. Exact minimaxity of selection and ranking rules under
symmetry is classical (\citealp{lehmann1966bahadur,eaton1967ranking}). We use it only for the exchangeable case.

\chg The literature on interference studies two-stage designs and their estimands (\citealp{hudgens2008toward}), optimal
saturation designs for power (\citealp{baird2018optimal}), policy learning with clustered or network interference
(\citealp{zhang2023individualized,park2024minimum,viviano2025policy}) and exposure mappings on networks
(\citealp{leung2020treatment}). These papers estimate or test effects, or \rev{learn rules that depend on covariates or
network data}. Our planner chooses among a few saturation policies for every cluster and uses \rev{neither covariates
nor network information}. A separate
literature studies inference on ranks and winners. \citet{mogstad2024inference} construct confidence sets for ranks,
\citet{andrews2024inference} correct inference on the best-performing treatment, and \citet{gu2023invidious} treat
ranking and selection as compound decision problems. The first two \rev{concern inference}, and the third
derives empirical Bayes ranking rules. We study the worst-case regret of a ranking rule for a fixed menu, and how to
design the experiment that \rev{generates its data}.

\chg All proofs and supporting material are in the Appendix.
\section{Setting}\label{sec:setting}

\subsection{Saturation policies and welfare}\label{sec:welfare}

\chg \rev{We consider a population of clusters, such as villages, schools or enumeration areas, to which a policy
applies.}
Cluster $i$ \rev{contains} $n_i$ units.
Unit $j$ in cluster $i$ has a potential outcome $Y_{ij}(\mathbf z)$ for every treatment vector
$\mathbf z=(z_1,\dots,z_{n_i})\in\{0,1\}^{n_i}$ of its own cluster, and it may depend on $\mathbf z$ in any way, for
instance through an unobserved network inside the cluster. \rev{Assumption~\ref{ass:model} below bounds the outcomes and
rules out interference between clusters.} We collect the primitives of cluster $i$ in its \emph{type}
$W_i=\bigl(n_i,\{Y_{ij}(\mathbf z):j\le n_i,\ \mathbf z\in\{0,1\}^{n_i}\}\bigr)$. As in \citet{manski2004statistical},
\rev{the potential outcomes of a given cluster are numbers, not random variables. All randomness arises from drawing
clusters from the population and from assigning treatment.}

\chg A saturation policy $\pi\in[0,1]$ treats exactly $n_i\pi$ units of each cluster, chosen by complete randomization.%
\footnote{Under complete randomization exactly $n_i\pi$ of the $n_i$ units are treated, and each of the
$\binom{n_i}{n_i\pi}$ sets of that size is equally likely. In a Bernoulli trial, each unit is instead treated
independently with probability $\pi$. The number treated is then $\mathrm{Binomial}(n_i,\pi)$ and equals $n_i\pi$ only
on average. A saturation policy fixes the share treated, so we use complete randomization. A similar analysis applies to
Bernoulli trials.}
Its assignment distribution is
\[
  \gamma(\pi,\mathbf z\mid n_i)=\binom{n_i}{n_i\pi}^{-1}\1\Bigl\{\textstyle\sum_{j=1}^{n_i}z_j=n_i\pi\Bigr\},
  \qquad \mathbf z\in\{0,1\}^{n_i}.
\]
The mean outcome of cluster $i$ under policy $\pi$ and the welfare of the policy are
\begin{equation}\label{eq:welfare}
  \Yb_i(\pi)=\frac{1}{n_i}\sum_{j=1}^{n_i}\sum_{\mathbf z\in\{0,1\}^{n_i}}Y_{ij}(\mathbf z)\,
  \gamma(\pi,\mathbf z\mid n_i),
  \qquad U(\pi)=\E_P\bigl[\Yb_i(\pi)\bigr],
\end{equation}
where $P$ is the superpopulation distribution of the cluster type. The distribution $\gamma$ describes the policy the
planner would
implement, not an experiment. $\Yb_i(\pi)$ is the \rev{marginal} group average potential outcome of
\citet{hudgens2008toward}.
Welfare weights clusters equally, which matches policies administered cluster by cluster and a design whose sampling
and assignment units are clusters.

\begin{assumption}\label{ass:model}
\chg (i) The potential outcomes take values in $[0,1]$ and depend only on the treatment vector of the unit's own cluster
\rev{(partial interference)}.
(ii) The planner ranks a finite menu $\Pib=\{\pi_1,\dots,\pi_K\}\subset[0,1]$ of distinct saturations, $2\le K<\infty$.
(iii) $n_i\pi_k$ is an integer for every $k$, $P$-almost surely.
\end{assumption}

\chg Part~(iii) makes the exact-count policy $\gamma(\pi_k,\cdot\mid n_i)$ well defined. Sections~\ref{sec:setting}
and~\ref{sec:finite} use it only \rev{to ensure that the experiment below assigns treatment with the same distribution
that defines welfare. These sections could instead use a rounding rule, such as nearest-integer treatment counts,
provided welfare is defined under the rounded policy and the experiment implements the same rule.}
Section~\ref{sec:optimality} also \rev{requires distinct saturations to have disjoint assignment supports, a property
that rounding may violate in small clusters.} We write $U_k=U(\pi_k)$,
$U=(U_1,\dots,U_K)$ and $\mathcal P$ for the set of distributions $P$ that satisfy Assumption~\ref{ass:model}.

\begin{example}\label{ex:models}
(a) A bounded linear-in-means reduced form: $Y_{ij}(\mathbf z)=a_{ij}+b_{ij}z_j+c_{ij}\bar z_{-j}+d_{ij}z_j\bar z_{-j}$,
where $\bar z_{-j}$ is the treated share among the other units of the cluster ($n_i\ge2$) and the coefficients keep
outcomes in $[0,1]$. Under the policy, $\E[z_j\bar z_{-j}]=\pi(n_i\pi-1)/(n_i-1)$, so welfare is at most quadratic in
$\pi$, and the best saturation can be interior. (b) Any
$[0,1]$-valued exposure model in which $Y_{ij}$ depends on $z_j$ and the treated share of $j$'s neighbours in an
unobserved network (\citealp{leung2020treatment}).
Neither model is assumed. Spillovers between clusters, such as through a market spanning villages, are excluded.
\hfill$\square$
\end{example}

\subsection{The randomized saturation experiment and the sampling model}\label{sec:experiment}

\chg The planner learns about $U$ from a two-stage randomized saturation experiment \rev{on $C$ sampled clusters}
(\citealp{hudgens2008toward,baird2018optimal}). A first stage assigns each sampled cluster to a saturation in $\Pib$. A
second stage treats units inside each cluster at the assigned saturation. Let $S_i\in\{1,\dots,K\}$ be the treatment
arm of cluster $i$, $\Sb=(S_1,\dots,S_C)$, and let $\Zb_i$ be the treatment vector of cluster $i$.

\begin{assumption}\label{ass:sampling}
(i) The types $W_1,\dots,W_C$ of the $C$ sampled clusters are i.i.d.\ draws from $P$. (ii) $\Sb$ is drawn by complete
randomization with fixed arm counts $C_1,\dots,C_K\ge1$, $\sum_kC_k=C$: it is uniform on the vectors with exactly $C_k$
entries equal to $k$, and it is independent of $(W_1,\dots,W_C)$. (iii) Given $\Sb$ and $(W_1,\dots,W_C)$, the vectors
$\Zb_1,\dots,\Zb_C$ are independent and $\Zb_i$ has distribution $\gamma(\pi_{S_i},\cdot\mid n_i)$.
\end{assumption}

\chg The counts $C_k$ are design choices made before the sample is drawn and are assumed to be given for now.
\rev{Section~\ref{sec:allocation} provides a rule for choosing them that minimizes a bound on the worst-case regret of
ES.} By Assumption~\ref{ass:sampling}(ii), each of the $C!/(C_1!\cdots
C_K!)$ ways to assign the $C$ clusters to the $K$ arms with exactly $C_k$ clusters in arm $k$ is equally likely. The
data are $(n_i,S_i,\Zb_i,\allowbreak Y_{i1},\dots,Y_{in_i})_{i\le C}$ with $Y_{ij}=Y_{ij}(\Zb_i)$. The network, if any,
is not observed. The estimator of $U_k$ averages the observed cluster means over the clusters in arm $k$:
\begin{equation}\label{eq:estimator}
  \Ybh_i=\frac{1}{n_i}\sum_{j=1}^{n_i}Y_{ij},\qquad
  \Uh_k=\frac{1}{C_k}\sum_{i:S_i=k}\Ybh_i,\qquad k=1,\dots,K.
\end{equation}

To state the independence result we index the observations by position within an arm. Let $\mathbf s=(s_1,\dots,s_C)$ be
a realization of $\Sb$. For each arm $k$, let $i_{k,1}(\mathbf s)<\dots<i_{k,C_k}(\mathbf s)$ be the labels of the
clusters with $s_i=k$, and let $X_{k,r}=\Ybh_{i_{k,r}(\Sb)}$ be the cluster mean in the $r$-th position of arm $k$. Then
$\Uh_k=C_k^{-1}\sum_{r=1}^{C_k}X_{k,r}$. Let $Q_k$ be the distribution of $n^{-1}\sum_{j\le n}Y_j(\mathbf Z)$ when the
type $(n,\{Y_j(\cdot)\})$ is drawn from $P$ and $\mathbf Z$ is then drawn from $\gamma(\pi_k,\cdot\mid n)$. $Q_k$ is
supported on $[0,1]$.

\begin{lemma}[Independent-cluster reduction]\label{lem:indep}
Let Assumptions~\ref{ass:model} and~\ref{ass:sampling} hold. For every assignment $\mathbf s$ in the support of $\Sb$,
conditional on $\Sb=\mathbf s$ the arm-position observations $\{X_{k,r}:k\le K,\ r\le C_k\}$ are mutually independent
and $X_{k,r}$ has distribution $Q_k$. Hence the array $(X_{k,r})$ is independent of $\Sb$ with distribution
$\bigotimes_{k}Q_k^{\otimes
C_k}$, the estimators $\Uh_1,\dots,\Uh_K$ are independent, and $\E[X_{k,r}]=\E[\Uh_k]=U_k$.
\end{lemma}

Complete randomization makes the $S_i$ dependent through the fixed arm counts, and the $Z_{ij}$ of a cluster
dependent through the fixed number of treated units (e.g.\ $n_i=2$, $\pi=1/2$ gives $Z_{i2}=1-Z_{i1}$). The $X_{k,r}$
are nevertheless independent. The first stage decides only which arm a cluster serves, so conditioning on $\Sb$ and
labelling clusters by their position within an arm removes the first dependence.\footnote{Label-indexed
cluster means are dependent through the count constraint: for $i\neq i'$,
$\mathrm{Cov}(\Ybh_i,\Ybh_{i'})=-\mathrm{Var}_\alpha(U_S)/(C-1)$, where
$\mathrm{Var}_\alpha(U_S)=\sum_k\alpha_kU_k^2-(\sum_k\alpha_kU_k)^2$ and $\alpha_k=C_k/C$.
$\Uh_k$ depends only on the arm-position observations.} Interference stays inside a
cluster, so each cluster mean depends only on that cluster's type and second-stage draw, and the second dependence is
absorbed, with all interference, into the bounded observation $X_{k,r}$ with unrestricted distribution $Q_k$.

As a result, the network never has to be modelled, observed or summarized, cluster size may be correlated with potential
outcomes, and ES depends on the data only through $K$ independent averages of i.i.d.\ observations in $[0,1]$ with means
$U_1,\dots,U_K$. Finite-sample bounds follow from inequalities for independent bounded variables, and the local analysis
of Section~\ref{sec:optimality} starts from the reduced experiment that keeps only these $K$ arm samples.

\subsection{Ranking rules and the loss}\label{sec:loss}

A ranking of the menu is a strict total order on $\{1,\dots,K\}$. We write it as the vector
$d=(d^{kk'})_{k<k'}\in\{0,1\}^T$, $T=K(K-1)/2$, with $d^{kk'}=1$ if $\pi_k$ is ranked above $\pi_{k'}$. The vectors that
arise in this way are the transitive ones.%
\footnote{That is, if $\pi_k$ is ranked above $\pi_{k'}$ and $\pi_{k'}$ above $\pi_{k''}$, then $\pi_k$ is
ranked above $\pi_{k''}$. A vector such as $d^{12}=d^{23}=1$, $d^{13}=0$ describes a cycle and is not a
ranking, so no ranking rule outputs it.}
A ranking rule $\delta=(\delta^{kk'})_{k<k'}$ maps the data to a ranking,
possibly with the help of an external randomization independent of the data. The empirical success (ES) ranking orders
the saturations by estimated welfare:
\begin{equation}\label{eq:esrule}
  \dES=(\dES^{kk'})_{k<k'},\qquad \dES^{kk'}=\1\{\Uh_k\ge\Uh_{k'}\},\qquad k<k'.
\end{equation}
The top-ranked saturation of a ranking is the one it ranks first. Ties are broken in favour of the lower index, so the
top-ranked saturation $\hat k$ of $\dES$ is the lowest index among the maximizers of $\Uh_k$.

\chg For a ranking $d$ and welfare vector $U$, the loss is
\begin{equation}\label{eq:loss}
  \Loss(d,U)=\sum_{k<k'}\Delta_{kk'}\Bigl[d^{kk'}\1\{U_k<U_{k'}\}+(1-d^{kk'})\1\{U_k>U_{k'}\}\Bigr],
  \qquad \Delta_{kk'}=|U_k-U_{k'}|.
\end{equation}
It is the welfare-weighted Kendall distance (\citealp{kendall1938new})\footnote{Losses that \rev{penalize} misordered
pairs
are
used to rank populations in decision problems. \citet{eaton1967ranking} studies a class of ranking losses that includes
this one, and \citet{kline2024discrimination} \rev{count} misordered pairs to grade employers.} between
$d$ and the true order: each misordered pair costs its welfare gap, and pairs with equal welfare cost nothing. The loss
also has a regret interpretation. The term for pair $(k,k')$ is the regret of a planner who uses the ranking as a
fallback when only $\pi_k$ and $\pi_{k'}$ turn out to be feasible, for example because of budget, equity or political
constraints. Moreover, if the feasible pair is drawn uniformly from the $T$ pairs, independently of the data, the
expected regret of choosing the higher-ranked member is $\Loss(d,U)/T$.

The risk of a rule $\delta$ under $P\in\mathcal P$ is its expected loss, which we call the ranking regret because the
loss is already a regret relative to the true order:
\begin{equation}\label{eq:risk}
  \Reg(\delta,P)=\E_P\bigl[\Loss(\delta,U)\bigr]=\sum_{k<k'}\Delta_{kk'}\,\P_P(\delta\text{ misorders }k,k'),
\end{equation}
where the expectation is over sampling, both stages of assignment and any external randomization, and a pair with
$\Delta_{kk'}>0$ is misordered when $\delta$ ranks the saturation with lower welfare above the other. The maximum risk
of $\delta$ over a class $\mathcal P_0\subseteq\mathcal P$ is $\sup_{P\in\mathcal P_0}\Reg(\delta,P)$, and the minimax
regret over $\mathcal P_0$ is $\inf_\delta\sup_{P\in\mathcal P_0}\Reg(\delta,P)$.

The next result shows that the loss also bounds the regret of choosing the top-ranked saturation.

\begin{proposition}\label{prop:dominate}
Let $d$ be any ranking of $\Pib$ with top-ranked saturation $\hat k$, and let $U_{k^\ast}=\max_kU_k$. Then
\[
  0\le U_{k^\ast}-U_{\hat k}\le\Loss(d,U).
\]
In particular, for every ranking rule $\delta$ and every $P$, the expected regret of the top-ranked saturation of
$\delta$ is at most its risk $\Reg(\delta,P)$.
\end{proposition}

The bound holds for every data set, so every upper bound on the ranking risk of ES below
also bounds the regret of selecting its top-ranked saturation, the treatment choice criterion of
\citet{manski2004statistical} and \citet{hirano2009asymptotics}.

\chg We have two remarks. First, \rev{ES satisfies independence of irrelevant alternatives (IIA).} Its ordering of
$\pi_k$
and $\pi_{k'}$ depends only on $(\Uh_k,\Uh_{k'})$\rev{. Adding saturations to the menu or removing other saturations
therefore does not change it.} The loss does change with the menu, because the task is to rank \rev{the saturations the
menu contains}. Second, the loss differs from the \rev{unweighted} count of misordered pairs, which \rev{treats the
reversal of two nearly equal saturations and the reversal of two widely separated ones alike}. Weighting each
misordered pair by its welfare gap makes \rev{the loss from a ranking error} match its consequence for the planner.
\section{Finite-sample regret bounds}\label{sec:finite}

This section gives a regret bound for the ES ranking that holds for every number of clusters. The bound uses no
information about the within-cluster network. It also shows how fast the risk falls as clusters are added.

We first bound the probability that ES reverses a pair of saturations and then sum these bounds over pairs to bound the
risk.

\begin{theorem}[Reversal and risk bound]\label{thm:finite}
Suppose Assumptions~\ref{ass:model} and~\ref{ass:sampling} hold. For every $P\in\mathcal P$: (i) if $U_k>U_{k'}$, then
\begin{equation}\label{eq:reversal}
  \P(\dES\text{ ranks }\pi_{k'}\text{ above }\pi_k)\;\le\;
  \exp\Bigl\{-\frac{2\Delta_{kk'}^2\,C_kC_{k'}}{C_k+C_{k'}}\Bigr\};
\end{equation}
(ii) the risk satisfies
\begin{equation}\label{eq:finite-bound}
  \Reg(\dES,P)\;\le\;\sum_{k<k'}\Delta_{kk'}\exp\Bigl\{-\frac{2\Delta_{kk'}^2\,C_kC_{k'}}{C_k+C_{k'}}\Bigr\}.
\end{equation}
By Proposition~\ref{prop:dominate}, the right side of~\eqref{eq:finite-bound} also bounds the expected regret of
choosing the top-ranked saturation.
\end{theorem}

\chg The proof applies the inequality of \citet{hoeffding1963probability} to $\Uh_k-\Uh_{k'}$, a difference of averages
of
independent observations in $[0,1]$ by Lemma~\ref{lem:indep}, as \citet{manski2004statistical} does for ES. What changes
is \rev{which observations are independent}: cluster means rather than individual outcomes, which are hard to treat as
independent when units interact. \rev{The bound therefore} holds whatever the network and the cluster sizes. With $m$
clusters in each of two arms,
the reversal probability is at most $\exp(-m\Delta^2)$, so $m\ge\log(1/\varepsilon)/\Delta^2$ clusters per arm order the
pair correctly with probability at least $1-\varepsilon$ (922 at $\Delta=0.05$ and $\varepsilon=0.1$).

\chg At a fixed distribution with all gaps positive and fixed shares $\alpha_k=C_k/C>0$, the bound
in~\eqref{eq:finite-bound}
decays exponentially in $C$, since $C_kC_{k'}/(C_k+C_{k'})=C\alpha_k\alpha_{k'}/(\alpha_k+\alpha_{k'})$. With a
vanishing share the exponent need not grow with $C$. Nor is the exponential rate uniform over distributions. For fixed
shares the uniform rate is $C^{-1/2}$. To state the uniform bound, define for shares $\alpha=(\alpha_1,\dots,\alpha_K)$
with all $\alpha_k>0$
\begin{equation}\label{eq:kappa}
\kappa_K(\alpha)=\sup_{h\in\R^K}\sum_{k<k'}
 \rev{|h_k-h_{k'}|}\exp\Bigl\{-\frac{2\rev{(h_k-h_{k'})^2}\,\alpha_k\alpha_{k'}}
  {\alpha_k+\alpha_{k'}}\Bigr\},
\end{equation}
where $h$ stands for the welfare vector scaled by $C^{1/2}$, $h_k=C^{1/2}U_k$, so that
\rev{$|h_k-h_{k'}|=C^{1/2}\Delta_{kk'}$}.
Define also the pairwise bound
\[
  \bar\kappa_K(\alpha)=(2\sqrt{e})^{-1}\sum_{k<k'}\{(\alpha_k+\alpha_{k'})/(\alpha_k\alpha_{k'})\}^{1/2},
\]
obtained by maximizing each pair's term over its own gap, as if the gaps were unrelated. Neither depends on $C$.

\begin{theorem}[Uniform rate]\label{thm:rate}
Suppose Assumptions~\ref{ass:model} and~\ref{ass:sampling} hold.
\begin{enumerate}
\item[(i)] (Upper bound.) For every $C$, with $\alpha_k=C_k/C$,
\[
  \sup_{P\in\mathcal P}\Reg(\dES,P)\le\min\bigl\{\kappa_K(\alpha)\,C^{-1/2},\,T/4\bigr\},\qquad
  \kappa_K(\alpha)\le\bar\kappa_K(\alpha)<\infty.
\]
\item[(ii)] (Balanced shares.) Among shares with all $\alpha_k>0$, both $\kappa_K$ and $\bar\kappa_K$ are minimized at
$\alpha_k=1/K$ (attainable when $K$ divides $C$), where $\bar\kappa_K=T\{K/(2e)\}^{1/2}$.
\item[(iii)] (Lower bound.) Along any sequence with $C\to\infty$ and $C_k/C\to\alpha_k>0$ for every $k$,
\[
  \liminf_{C\to\infty}C^{1/2}\sup_{P\in\mathcal P}\Reg(\dES,P)\ge\frac{(K-1)\,\astar}{2\min_k\alpha_k^{1/2}},\qquad
  \astar=\max_{u\ge0}u\Phi(-u),
\]
where $\Phi$ is the standard normal distribution function. Hence the order $C^{-1/2}$ in~(i) is exact for ES.
\end{enumerate}
\end{theorem}

\chg Theorem~\ref{thm:rate} has three \rev{implications}. First, the worst-case ranking regret of ES falls at rate
$C^{-1/2}$ for
every number of clusters, whatever the interference within clusters, with a constant that depends only on the arm
shares. Requiring the gaps to come from one welfare vector makes this constant smaller than the pairwise bound (Appendix~\ref{sec:S2} gives values). Second, balanced shares minimize the constant. Third, the rate is sharp for ES,
so the bound describes how fast its regret actually falls.

Cluster size does not appear in these bounds: if units in a cluster move together, a cluster of any size carries no more
information than one unit, so a distribution-free bound cannot reward size. Larger clusters help through a smaller
variance of the cluster mean, which Section~\ref{sec:optimality} uses. The superpopulation target also needs no
inequalities for sampling without replacement, because the cluster means within each arm are independent.

These bounds use only the range of the outcomes, so they are conservative and leave open whether another rule could do
much better. Section~\ref{sec:optimality} answers this question in a local asymptotic framework.
% Section 4: near-minimax optimality and the allocation of clusters.

\section{Near-minimax optimality and the allocation of clusters}\label{sec:optimality}

This section asks whether any rule can do substantially better than the ES ranking. In a local asymptotic framework, no
rule can lower the worst-case regret of ES on efficient welfare estimates by more than a factor of about 1.4, for any
number of saturations. Based on the bound, we study an allocation rule for clusters.

\subsection{Local asymptotics}

The number of clusters grows while the menu stays fixed. Regret is largest when welfare differences are comparable to
sampling error, so we study welfare vectors close to a tie, $U=u_0\1+h/\sqrt{C}$ with $h\in\R^K$ unrestricted, as in
\citet{hirano2009asymptotics}.

\begin{assumption}[Local regime]\label{ass:local}
(i) $C\to\infty$ with $K$ fixed and $C_k/C\to\alpha_k\in(0,1)$ for every $k$. (ii) Cluster sizes take values in a
fixed finite set $\mathcal N$ on which Assumption~\ref{ass:model} holds, with unknown probabilities $p_n>0$. (iii) At
the reference distribution $P^0$, $U_k(P^0)=u_0$ for every $k$, and the within-group variance
$w_k=\E_{P^0}\{\operatorname{Var}_{P^0}(\Ybh_i\mid S_i=k,n_i,\Zb_i)\}$ is positive for every $k$.
\end{assumption}

The data of cluster $i$ are $O_i=(S_i,n_i,\Zb_i,Y_{i1},\dots,Y_{in_i})$. The argument of Lemma~\ref{lem:indep}
applies to the whole observation $O_i$: given $\Sb$, the observations in the arm positions are independent, and an
observation in arm $k$ has a distribution $P_k$ that does not depend on $\Sb$. Under Assumption~\ref{ass:sampling},
$P_k$ has
the form
$p_n\,\gamma(\pi_k,z\mid n)\,F_{knz}(dy)$: the size distribution $(p_n)$ is common to all arms, the assignment
probabilities
$\gamma(\pi_k,z\mid n)$ are known, and the conditional distribution $F_{knz}$ of the outcome vector is unrestricted. Let
$\mathcal P_{\mathrm{loc}}$ be the set of arm-distribution tuples $P=(P_1,\dots,P_K)$ of this form. Each is induced by
some type
distribution in $\mathcal P$, because distinct saturations use disjoint assignment vectors. Write
$\mu_k(n,z)=\E_P(\Ybh_i\mid S_i=k,n_i=n,\Zb_i=z)$ and $m_k(n)=\sum_z\gamma(\pi_k,z\mid n)\mu_k(n,z)$, so that
$U_k=\sum_np_nm_k(n)$. The \emph{reduced experiment} observes only $(S_i,\Ybh_i)$, with unrestricted arm distributions.

Cluster size and the assignment pattern are observed and their distributions are partly known, so raw arm means
need not be efficient in the full experiment. The adjusted estimator
\begin{equation}\label{eq:Ueff}
 \Ueff_k=\sum_{n\in\mathcal N}\widehat p_n\sum_{z}\gamma(\pi_k,z\mid n)\,\widehat\mu_k(n,z)
\end{equation}
uses this information: $\widehat p_n$ is the share of size-$n$ clusters among \emph{all} $C$ clusters, and
$\widehat\mu_k(n,z)$ is the mean of $\Ybh_i$ over arm-$k$ clusters of size $n$ with assignment $z$. It averages the
cluster means of arm $k$ within each cluster size and assignment pattern, then reweights these averages by the size
shares, pooled over all clusters, and the known assignment probabilities. In survey sampling this is poststratification.
If a group is empty we set $\widehat\mu_k(n,z)=0$, an event whose probability vanishes under Assumption~\ref{ass:local}.
Let $\deff$ rank $\Ueff$, with ties broken by the lower index, as for $\dES$. We call $\dES$, which ranks the raw arm
means $\Uh_k$, raw ES. The next lemma identifies the efficient covariance and compares it with that of the raw arm
means. Recall that $m_k(n)=\sum_z\gamma(\pi_k,z\mid n)\mu_k(n,z)$, write $m(n)=(m_1(n),\dots,m_K(n))'$, and let
$\operatorname{Cov}_p$ denote the covariance over cluster sizes $n_i$ drawn from $(p_n)$.

\begin{lemma}[Efficient welfare estimation]\label{lem:efficient}
Under Assumptions~\ref{ass:model}, \ref{ass:sampling} and~\ref{ass:local}, $\sqrt C\{\Ueff-U(P^0)\}\to\Normal(0,V)$
in distribution under the reference distribution $P^0$, where
\begin{equation}\label{eq:effcov}
 V=\operatorname{Cov}_p\{m(n_i)\}+\operatorname{diag}(w_k/\alpha_k),
\end{equation}
and $V$ is the efficiency bound for $U$ in $\mathcal P_{\mathrm{loc}}$. The raw arm means have asymptotic covariance
$V_{\mathrm{raw}}=\operatorname{diag}(v_k/\alpha_k)\succeq V$, with $v_k=\operatorname{Var}_{P^0}(\Ybh_i\mid S_i=k)$.
They are efficient in the reduced experiment, and $V=V_{\mathrm{raw}}$ when $\mu_k(n,z)=U_k$ for every supported
$(k,n,z)$.
\end{lemma}

The raw arm means are as precise as $\Ueff$ only when neither cluster size nor the assignment pattern predicts the
cluster mean. The number of assignment patterns grows quickly with $n$, so in practice one adjusts for cluster size or a
coarse function of $z$, which attains $V$ when $\mu_k(n,z)$ depends on $z$ only through it.

\chg For a positive definite $K\times K$ matrix $V$, consider the limit experiment of observing $X\sim\Normal(h,V)$ with
$h\in\R^K$ unrestricted. A rule maps $X$ to a (possibly randomized) ranking, its risk is
$\Reg_V(\delta,h)=\E_h\Loss(\delta(X),h)$ with $\Loss$ the loss \eqref{eq:loss}, and
$\Reg^*(V)=\inf_\delta\sup_h\Reg_V(\delta,h)$. Let $s_{kk'}(V)=\{(e_k-e_{k'})'V(e_k-e_{k'})\}^{1/2}$ be the standard
deviation of $X_k-X_{k'}$, where $e_k$ is the \rev{$k$-th} unit vector in $\R^K$, and let $Q(V)=\sum_{k<k'}s_{kk'}(V)$.
The ES
ranking in the limit experiment orders the saturations by the components of $X$. Its maximum risk is
\[
 \MES(V)=\sup_{h\in\R^K}\sum_{k<k'}|h_k-h_{k'}|\,\Phi\!\left(-\frac{|h_k-h_{k'}|}{s_{kk'}(V)}\right).
\]

To state the result for the cluster experiment, let $H^2(F,G)=\int(\sqrt{dF}-\sqrt{dG})^2$ and
$\Hnb_C(b)=\{P\in\mathcal P_{\mathrm{loc}}:\sum_k\alpha_kH^2(P_k,P_k^0)\le b^2/C\}$. Write $\Reg(\delta_C,P)$ for the
risk \eqref{eq:risk} of a rule under $P$, with loss evaluated at $U(P)$. The local asymptotic minimax regret is
\[
 \Reg^{\mathrm{loc}}_*=\lim_{b\to\infty}\liminf_{C\to\infty}\inf_{\delta_C}
  \sup_{P\in\Hnb_C(b)}\sqrt C\,\Reg(\delta_C,P),
\]
where the infimum is over all ranking rules based on the full data $O_1,\dots,O_C$, including randomized ones.
Similarly, we can define the local maximum regret of a rule sequence $\delta_C$ by the same expression with $\delta_C$
fixed and $\limsup$ in place of $\liminf$.

The next theorem bounds the local asymptotic minimax regret below and above by two multiples of $Q(V)$, and shows that
the ranking of efficient welfare estimates is within the upper one.

\begin{theorem}[Near-minimaxity]\label{thm:nearminimax}
Let Assumptions~\ref{ass:model}, \ref{ass:sampling} and \ref{ass:local} hold, and let $V$ be given by
\eqref{eq:effcov}.
\begin{enumerate}
\item[(i)] (Lower bound.) There is a submodel $\{P_{C,h}:h\in\R^K\}\subset\mathcal P_{\mathrm{loc}}$ through $P^0$
 with $\sqrt C\{U(P_{C,h})-u_0\1\}\to h$ such that
 \[
  \sup_{H\subset\R^K,\ |H|<\infty}\ \liminf_{C\to\infty}\ \inf_{\delta_C}\ \max_{h\in H}\
  \sqrt C\,\Reg(\delta_C,P_{C,h})\ \ge\ \bstar Q(V),
 \]
 where $\bstar=(2\pi)^{-1/2}\{\tau-\tau^2(1+\tau^2)^{-1/2}\}$ with $\tau^2=(\sqrt5-1)/2$.
 Consequently $\Reg^{\mathrm{loc}}_*\ge\bstar Q(V)$.
\item[(ii)] (Upper bound.) For every $b>0$,
 $\limsup_{C}\sup_{P\in\Hnb_C(b)}\sqrt C\,\Reg(\deff,P)\le\MES(V)\le\astar Q(V)$, where
 $\astar=\max_{u\ge0}u\Phi(-u)$.
\item[(iii)] (Factor.) $\bstar Q(V)\le\Reg^{\mathrm{loc}}_*\le\astar Q(V)$, and the local maximum regret of $\deff$
 is at most $(\astar/\bstar)\Reg^{\mathrm{loc}}_*$, where $\astar/\bstar=1.41884<1.42$ for every $K$ and every $V$.
\end{enumerate}
\end{theorem}

The upper bound $\astar Q(V)$ sums the two-action minimax value of \citet{hirano2009asymptotics} over pairs. The new
part is the lower bound. Pairwise welfare differences are not independent parameters. For example,
$h_1-h_3=(h_1-h_2)+(h_2-h_3)$. Therefore, a prior must be specified over the welfare vector $h$, rather than
independently over all pairwise differences, to ensure that the implied contrasts are mutually consistent. We use the
prior $h\sim\Normal(0,\tau^2V)$. Its posterior mean $\tau^2X/(1+\tau^2)$ is a positive multiple of $X$, so ES is Bayes,
with Bayes risk $\bstar Q(V)$ at $\tau^2=(\sqrt5-1)/2$. The factor is a worst-case bound: the ratio of the local maximum
regret of $\deff$ to the minimax regret equals one with two saturations or an exchangeable covariance.

In the reduced experiment, where only the arm labels and cluster means are observed, raw ES coincides with efficient ES.
Hence, Theorem~\ref{thm:nearminimax} applies directly to raw ES, with the covariance matrix corresponding to the reduced
experiment.

\chg In the full experiment, raw ES has local maximum regret bounded by $\astar Q(V_{\mathrm{raw}})$. However, the
near-minimax guarantee in part~(iii) does not generally apply because $V_{\mathrm{raw}}$ may exceed the efficiency bound
$V$. The guarantee continues to hold if $\mu_k(n,z)=U_k$ for every arm, \rev{cluster size and} assignment pattern. Under
this
condition, the additional information does not improve asymptotic efficiency, so $V_{\mathrm{raw}}=V$.

\begin{corollary}[Two saturations]\label{cor:twosat}
Let $K=2$ and suppose the assumptions of Theorem~\ref{thm:nearminimax} hold. Then ES on the efficient welfare estimates
is locally asymptotically minimax, with minimax value $\Reg^{\mathrm{loc}}_*=\astar s_{12}(V)$.
\end{corollary}

This is the two-action result of \citet{hirano2009asymptotics} in the cluster experiment. Raw ES has the same property
when $\mu_k(n,z)=U_k$ in every group. The Gaussian bound $\bstar Q$ is not sharp here, because the least favourable
prior has two points, $h=\pm(u^*/s_{12})V(e_1-e_2)$ with $u^*$ the maximizer of $u\Phi(-u)$. ES is also exactly minimax
when the covariance is exchangeable, $V=aI+b\1\1'$ with $a>0$, a consequence of \citet{eaton1967ranking}
(Proposition~\ref{prop:S6-exch}).

With three or more saturations, the upper bound is strict: $\Reg^*(V)\le\MES(V)<\astar Q(V)$ for every $V$, because the
pairwise worst cases cannot occur together (Lemma~\ref{lem:S6-strict}). For general $V$ the
minimax value lies between $\bstar Q(V)$ and $\MES(V)$, and we do not characterize it.

\subsection{Allocating clusters}\label{sec:allocation}

\chg So far the design, the arm counts $C_k$, has been taken as given. This subsection asks how the designer should
choose
them before the data arrive. Theorem~\ref{thm:nearminimax} suggests a criterion: choose the counts to minimize the upper
bound on the worst-case regret of ES. We work in the reduced experiment, in which only the arm and the observed mean of
each cluster are recorded, so the arms are independent and $V=\operatorname{diag}(v_k/\alpha_k)$, where $v_k$ is the
variance of a cluster mean in arm $k$. With $m_k$ clusters in arm $k$, the estimated gap between $\pi_k$ and $\pi_{k'}$
has standard deviation $(v_k/m_k+v_{k'}/m_{k'})^{1/2}$, and the upper bound on the worst-case regret of ES is $\astar
G(m;v)$, where
\[
 G(m;v)=\sum_{k<k'}\Bigl(\frac{v_k}{m_k}+\frac{v_{k'}}{m_{k'}}\Bigr)^{1/2}
\]
is \rev{the sum of the standard deviations of the estimated pairwise welfare contrasts}. With per-cluster costs $c_k>0$
and budget $B>0$, the designer
minimizes $G(m;v)$ subject to $\sum_kc_km_k=B$. The efficient covariance \eqref{eq:effcov} couples the arms through
cluster size, and the corresponding problem would need a separate derivation.

\begin{theorem}[Allocation]\label{thm:allocation}
Suppose $v_k,c_k>0$ for every $k$ and $B>0$.
\begin{enumerate}
\item[(i)] $G(\cdot\,;v)$ is strictly convex on $\{m>0:\sum_kc_km_k=B\}$ and has a unique interior minimizer $m^*$,
 characterized by
 \begin{equation}\label{eq:alloc-foc}
  \frac{v_k}{c_k\,m_k^{2}}\sum_{k'\ne k}\Big(\frac{v_k}{m_k}+\frac{v_{k'}}{m_{k'}}\Big)^{-1/2}=\lambda,
  \qquad k=1,\dots,K,
 \end{equation}
 for some $\lambda>0$.
\item[(ii)] If $v_k=v$ and $c_k=c$ for all $k$, then $m_k^*=B/(cK)$. For $K=2$, $m_1^*/m_2^*=\{v_1c_2/(v_2c_1)\}^{1/2}$.
\item[(iii)] For every $m$, the maximum of $G(m;v)$ over $v\in(0,1/4]^K$ (cluster means in $[0,1]$) is attained at
 $v_k=1/4$ for every $k$. Hence the robust allocation solves \eqref{eq:alloc-foc} with $v_k=1/4$ and is balanced when
 costs are equal.
\end{enumerate}
\end{theorem}

\chg Theorem~\ref{thm:allocation} characterizes the unique allocation that minimizes the Gaussian regret upper bound.
The
first-order conditions equate marginal reductions in the bound per unit cost, so clusters move toward saturations whose
cluster means vary more and toward cheaper ones. With two saturations and equal costs the solution is the Neyman
allocation, which \citet{adusumilli2026how} shows is minimax-regret optimal for choosing between two treatments. For
$K\ge3$ it is computed numerically. When variances are unknown and costs are equal, balanced allocation is robustly
optimal. \rev{The allocation minimizes the bound rather than the minimax regret. Because the bound is within the factor
$\astar/\bstar$ of the minimax regret at every allocation, ES under $m^*$ has maximum regret less than 1.42 times the
smallest minimax regret that any budget-feasible allocation attains with any ranking rule in the reduced experiment,
although $m^*$ need not be exactly minimax.}

In practice, cluster counts must be integers and may have minimums. Under equal variances and costs, the optimal integer
allocations differ by at most one (Proposition~\ref{prop:S2-integer}). More generally, rounding
down an interior $m^*$ with $\min_km_k^*>1$ keeps the budget feasible and raises $G$ by at most the factor
$(1-1/\min_km_k^*)^{-1/2}$, since $G$ is decreasing in each count and homogeneous of degree $-1/2$. Minimum-count
constraints preserve convexity, and pilot data can supply planning values for $v_k$. Unlike \citet{baird2018optimal},
who optimize power for specified effects, our criterion weighs all pairwise comparisons. Section~\ref{sec:sim} evaluates
the bounds in finite samples.
\section{Simulation evidence}\label{sec:sim}

\chg We check whether the theory of Sections~\ref{sec:finite} and~\ref{sec:optimality} describes ranking regret at the
numbers of clusters used in saturation experiments: the rate and the local approximation, the size of the finite-sample
bounds, and the allocation rule. The menu is $\Pib=\{0.1,0.3,0.5,0.7\}$, so $K=4$, and the design is balanced with
$C=4m$
clusters, $m\in\{2,4,\rev{8,}\dots,64\}$. Every design uses $M=10^6$ Monte Carlo replications. We compare three
quantities: the
simulated risk, the Gaussian benchmark $\MES(V)$ that the local theory predicts for large $C$, and the finite-sample
bounds of Theorems~\ref{thm:finite} and~\ref{thm:rate}, which hold for every $C$. Full designs and tables are in Appendix~\ref{sec:S4}.

\emph{The local approximation is accurate at small $C$.} In the Beta-Binomial design a cluster in arm $k$ draws
$P_i\sim\mathrm{Beta}(qU_k,q(1-U_k))$ and independent Bernoulli$(P_i)$ unit outcomes, with intracluster correlation
(ICC) $1/(q+1)$ equal to 0.05 or 0.25. Local alternatives set $U_k=1/2+h_k/\sqrt C$, with $h$ on a finite grid $H$.
Figure~\ref{fig:sim}(a) plots $\sqrt C$ times the largest simulated risk over $H$, a lower bound on the worst case (a
finer grid changes it by at most 0.010, Table~\ref{tab:S4-rate}). All arms have the same variance, so $V$ is
exchangeable, ES is exactly minimax, and $\MES(V)$ is the local minimax value (Section~\ref{sec:optimality}). The curves
are flat from $C=8$, at 0.40 and 0.67, against $\MES(V)=0.399$ and 0.679. The picture is consistent with the rate
$C^{-1/2}$ of Theorem~\ref{thm:rate} and the local value of Theorem~\ref{thm:nearminimax}: with eight clusters the risk
at grid-selected adverse welfare vectors is already close to the Gaussian benchmark.

\chg \emph{The finite-sample bounds are valid but conservative.} The uniform bound of Theorem~\ref{thm:rate} is six to
eleven
times the worst-case regret in Figure~\ref{fig:sim}(a). At fixed welfare gaps, Figure~\ref{fig:sim}(b) plots the bound
of Theorem~\ref{thm:finite} divided by the simulated risk. The ratio is about 3 at $C=8$ and rises to between 11 and 39
at $C=256$, as the largest standardized gap $\Delta_{kk'}/\{(v_k+v_{k'})/m\}^{1/2}$ grows from \rev{at most 1.2} to
between 3 and 7, because
the bound uses only the range of the outcomes (Table~\ref{tab:S4-bound}). Jointly over all 78 comparisons, at the 5\%
level, no simulated risk is significantly above its bound (Appendix~\ref{sec:S4-mc}).

\emph{Unobserved interference and allocation.} In two network designs an unobserved stochastic block model links the
units of a cluster and a unit's outcome depends on its own treatment and its neighbours' treated share. The risk of ES
is within 2.5\% of a normal approximation that uses only the welfare gaps and the variances of the cluster means
(Table~\ref{tab:S4-network}). With 60 clusters and arms whose ICC ranges from 0.05 to 0.5, the allocation of
Theorem~\ref{thm:allocation} lowers the largest regret over the grid by 6.2\% relative to an equal split, while tilting
the counts toward the quieter arms raises it by 15.3\% (Table~\ref{tab:S4-alloc}).

\begin{figure}[t]
\centering
\includegraphics[width=\textwidth]{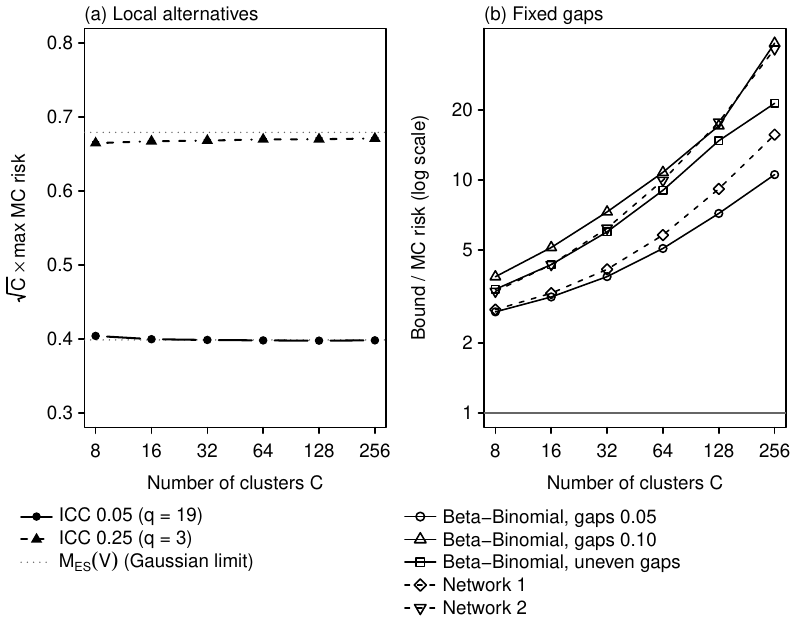}
\caption{Regret of the ES ranking}\label{fig:sim}
\vspace{4pt}
\begin{minipage}{\textwidth}
\footnotesize
\emph{Notes:} $K=4$, balanced design with $C=4m$ clusters. Panel (a): Beta-Binomial design, $U_k=1/2+h_k/\sqrt C$,
$\sqrt C$ times the largest simulated risk over $H$, with 95\% Monte Carlo intervals at the maximizing $h$ from
independent replications. Dotted lines: $\MES(V)$ for each ICC. Panel (b): fixed welfare gaps, Beta-Binomial design with
ICC 0.25 and both network designs, bound of Theorem~\ref{thm:finite} (capped at $T/4$) divided by the simulated risk.
Values above one mean the bound holds.
\end{minipage}
\end{figure}
\section{Application: the Zomba cash-transfer saturation experiment}\label{sec:application}

\chg We apply the ES ranking and the bounds of Sections~\ref{sec:finite} and~\ref{sec:optimality} to the cash-transfer
experiment of \citet{baird2011cash} in Zomba, Malawi. The experiment sampled 176 enumeration areas (EAs), 88 programme
EAs and 88 pure controls, randomized within three geographic strata. In each programme EA a nominal share 0, 1/3, 2/3 or
1 of the sampled baseline schoolgirls was offered a conditional or an unconditional cash transfer (CCT and UCT,
respectively), with the \rev{recipients drawn at random} within the EA.\footnote{Replication data of
\citet{baird2011cash},
Harvard Dataverse, \url{https://doi.org/10.7910/DVN/7K1A25}. Sample construction is in Appendix~\ref{sec:S5}.}
\rev{Under a CCT, the monthly payment to the girl and to her parents or guardians was withheld if she had attended
school
on fewer than 80\% of the school days in the previous month. A UCT offer made \rev{comparable} monthly payments with no
schooling condition (\citealp{baird2011cash}).}
With the EA as the cluster, this is the two-stage design of Assumption~\ref{ass:sampling} apart from blocking, and each
transfer type gives a menu of $K=4$ saturations. The outcomes are school enrolment, never married and never pregnant at
the third survey round, averaged over the sampled girls of an EA. An offer to one girl can change the outcomes of others
through peers at school and the local marriage market, which is why the design varies the share offered. The zero
saturation of both menus is the 15 programme EAs with no schoolgirl offers\rev{, in which only girls who had dropped
out of school before baseline were offered CCTs}. Using it for UCTs assumes that the two transfer types are equivalent
when \rev{no schoolgirl} is offered, and the pure controls belong to neither menu. Welfare is the
cluster-weighted mean outcome under the field's rule, which offers the nearest integer to $\pi_kn_i$ at random, and we
estimate it by the arm means.

\chg Table~\ref{tab:zomba} reports $C_k$, $\Uh_k$ with standard errors and the ES ranking. For UCTs and \rev{the
never-married outcome} the
ranking is monotone, $1,2/3,1/3,0$. Full saturation \rev{exceeds} zero by 0.129 ($z=2.14$), the only standardized gap
above 2 in
either menu, and ``1 first, 0 last'' holds under sampling weights and in the \rev{panel subsample (girls with
\texttt{panel}${}=1$)}. The two intermediate
saturations are not separated ($z=-0.16$), and for never pregnant $2/3$ and $1$ are nearly tied at the top. The CCT
rankings are non-monotone, with every standardized gap at most 1.17 in absolute value. The data thus favour full over
zero UCT saturation for delaying marriage, but they do not order the intermediate saturations.

\begin{table}[t]
\centering
\caption{\chg ES rankings of saturations in the Zomba experiment}\label{tab:zomba}
\small\setlength{\tabcolsep}{2.5pt}
\begin{tabular*}{\textwidth}{@{\extracolsep{\fill}}llccccc@{}}
\toprule
\rev{Transfer} & Outcome & $0$ & $\tfrac13$ & $\tfrac23$ & $1$ & ES ranking \\
\midrule
CCT & EAs, $C_k$ & 15 & 15 & 16 & 15 & $C=61$ \\
    & Enrolled       & 0.614 (0.052) & 0.595 (0.036) & 0.656 (0.049) & 0.608 (0.036)
    & $\tfrac23$, $0$, $1$, $\tfrac13$ \\
    & Never married  & 0.754 (0.046) & 0.804 (0.035) & 0.766 (0.049) & 0.742 (0.040)
    & $\tfrac13$, $\tfrac23$, $0$, $1$ \\
    & Never pregnant & 0.715 (0.042) & 0.664 (0.066) & 0.682 (0.051) & 0.671 (0.034)
    & $0$, $\tfrac23$, $1$, $\tfrac13$ \\
\midrule
UCT & EAs, $C_k$ & 15 & 9 & 9 & 9 & $C=42$ \\
    & Enrolled       & 0.614 (0.052) & 0.625 (0.036) & 0.589 (0.054) & 0.669 (0.038)
    & $1$, $\tfrac13$, $0$, $\tfrac23$ \\
    & Never married  & 0.754 (0.046) & 0.807 (0.040) & 0.818 (0.055) & 0.883 (0.039)
    & $1$, $\tfrac23$, $\tfrac13$, $0$ \\
    & Never pregnant & 0.715 (0.042) & 0.724 (0.042) & 0.777 (0.057) & 0.770 (0.033)
    & $\tfrac23$, $1$, $\tfrac13$, $0$ \\
\bottomrule
\end{tabular*}
\par\smallskip
\begin{minipage}{\textwidth}
\footnotesize
\emph{Notes:} Clusters are EAs. Saturations are nominal shares of sampled baseline schoolgirls offered a transfer.
Entries are $\Uh_k$ with standard errors in parentheses, from the sample variances $s_k^2$ of EA means within a
saturation. The 0 saturation is the 15 programme EAs with no schoolgirl offers, shared by both menus. Its use in the
UCT rows assumes policy equivalence. The ES ranking lists saturations from best to worst.
\end{minipage}
\end{table}

\chg How much welfare could a planner lose by following these rankings? Because the EA means within an arm are
independent
whatever the spillovers (Lemma~\ref{lem:indep}), the bounds of Sections~\ref{sec:finite} and~\ref{sec:optimality} apply
without a model of the network. Plugging the sample variances $\hat v_k=s_k^2$ of the EA means into the bound of
Theorem~\ref{thm:nearminimax} gives a worst-case regret for raw ES of $\astar\widehat G=0.061$ to 0.071 at the realized
counts, with $\widehat G(m)=G(m;\hat v)$. Summed over the six pairs, the expected welfare cost of misordering is thus at
most about 6 to 7 percentage points of the outcome share, and by Theorem~\ref{thm:nearminimax} no rule that uses only
the arm labels and the EA means could lower this worst case by more than a factor of 1.42. These values are a local
asymptotic approximation with estimated variances, not a finite-sample guarantee. The distribution-free bound of
Theorem~\ref{thm:rate} holds for every outcome distribution. Because it uses only the range of the outcomes\rev{,} it is
larger
here, 0.553 for CCTs and 0.692 for UCTs with the numerically computed constant $\kappa_4$.

Could the 88 programme EAs have been allocated better? We minimize $\widehat G_{\mathrm{CCT}}+\widehat G_{\mathrm{UCT}}$
over allocations with at least one EA per saturation, summing the two because the zero arm enters both menus. The best
integer allocation lowers the bound by only 2.3 to 3.6\% (1.1 to 8.4\% allowing for the uncertainty in $\hat v_k$,
Appendix~\ref{sec:S5}), so the realized design is close to minimizing the plug-in bound. The optimum adds EAs
to the shared zero arm, 17 to 22 instead of 15, and takes them mainly from the CCT arms, trading some CCT precision (its
bound rises 2.6 to 5.5\%) for UCT precision (its bound falls 7.8 to 10.2\%). Under balance, $\astar\widehat G$ reaches
0.05 at 17 to 31 EAs per saturation. The bound falls at rate $C^{-1/2}$, so halving it needs four times the clusters,
352 instead of 88 programme EAs. Adding clusters matters more than reallocating them here.

Two features of the experiment are not exactly covered by Assumption~\ref{ass:sampling}. First, the first stage was
blocked on three strata. Conditional on the stratum labels the EA means within each arm remain independent and bounded,
so the bounds apply to the arm targets $\E\,\Uh_k=\sum_s(C_{k,s}/C_k)U_k^{(s)}$, whose stratum weights differ across
arms. The difference is largest for UCTs, whose zero arm includes three far-rural EAs while saturations 1/3 and 1 have
none. Second, two programme EAs have offer counts one away from the nearest-integer rule, and dropping them changes no
ES ranking. Attrition is 4.4\%, and we do not reject equal attrition across saturations.
\section{Conclusion}\label{sec:conclusion}

A planner who chooses among saturations often needs their order, not only the best one, and the ranking loss also bounds
the regret of the top-ranked saturation. Without any model of within-cluster interference, the ES ranking has
finite-sample regret bounds, and for fixed allocation shares its worst-case regret falls at the exact order $C^{-1/2}$.
Ranking efficient welfare estimates is locally near-minimax for every number of saturations and every covariance, and
exactly minimax with two saturations. The upper bound on its regret also gives a variance- and cost-sensitive allocation
of clusters. Natural extensions are spillovers between clusters, adaptive or multi-wave designs that spend later
clusters on the pairs left unresolved, saturation rules that depend on cluster covariates, and the exact minimax ranking
rule and allocation under the efficient covariance.

\bibliography{reference}

% ===================================================================================================================
% Appendix: proofs (Appendix A of main.tex) followed by the supplement sections (supplement.tex)
% ===================================================================================================================
\clearpage
\appendix
\counterwithin*{table}{section}
\counterwithin*{figure}{section}
\renewcommand{\thetable}{\thesection.\arabic{table}}
\renewcommand{\thefigure}{\thesection.\arabic{figure}}
\begin{center}
{\Large Appendix}
\end{center}
\addcontentsline{toc}{section}{Appendix}
\medskip
\noindent
Appendix~\ref{app:s4} proves the results stated in the main text. Appendix~\ref{sec:S1} gives the uniformity details
and the finite-support approximation used in the proof of Theorem~\ref{thm:nearminimax}. Appendix~\ref{sec:S2}
computes the constant $\kappa_K$ of Theorem~\ref{thm:rate} and proves the integer-allocation statement of
Section~\ref{sec:optimality}. Appendix~\ref{sec:S4} gives the simulation designs and full tables,
Appendix~\ref{sec:S5} the construction of the application data and the sensitivity checks of
Section~\ref{sec:application}, and Appendix~\ref{sec:S6} the strictness of the upper bound stated after
Theorem~\ref{thm:nearminimax} and the exchangeable case.

\section{Proofs}\label{app:s4}
\medskip\textbf{Proof of Lemma~\ref{lem:indep}.}
Fix $\mathbf s$ in the support of $\Sb$. By Assumption~\ref{ass:sampling}(ii), given $\Sb=\mathbf s$ the types
$W_1,\dots,W_C$ are still i.i.d.\ with distribution $P$, and by Assumption~\ref{ass:sampling}(iii) the $\Zb_i$ are then
independent with distributions
$\gamma(\pi_{s_i},\cdot\mid n_i)$, well defined by Assumption~\ref{ass:model}(iii). So the pairs $(W_i,\Zb_i)$ are
independent given $\Sb=\mathbf s$. By Assumption~\ref{ass:model}(i), $\Ybh_i=g(W_i,\Zb_i)$ with
$g((n,\{Y_j(\cdot)\}),\mathbf z)=n^{-1}\sum_{j\le n}Y_j(\mathbf z)$, the same $g$ for every cluster, so the $\Ybh_i$
are independent given $\Sb=\mathbf s$ and $\Ybh_i$ has distribution $Q_{s_i}$. Given $\Sb=\mathbf s$, $(k,r)\mapsto
i_{k,r}(\mathbf s)$ is a deterministic bijection onto $\{1,\dots,C\}$ with $s_{i_{k,r}(\mathbf s)}=k$, so the
$X_{k,r}=\Ybh_{i_{k,r}(\mathbf s)}$ are a relabelling of independent variables, with joint distribution
$\bigotimes_kQ_k^{\otimes C_k}$ free of $\mathbf s$. Averaging over $\Sb$ gives the unconditional distribution and the
independence from $\Sb$, and each $\Uh_k$ is a function of its own block of the array.
\[
  \E[X_{k,r}]
  =\E_P\Bigl[\frac1n\sum_{j=1}^{n}\sum_{\mathbf z\in\{0,1\}^{n}}Y_j(\mathbf z)\,\gamma(\pi_k,\mathbf z\mid n)\Bigr]
  \overset{\eqref{eq:welfare}}{=}U(\pi_k).
  \tag*{$\square$}
\]

\medskip\textbf{Proof of Proposition~\ref{prop:dominate}.}
By definition of $k^\ast$, $U_{k^\ast}-U_{\hat k}\ge0$. If $U_{\hat k}=U_{k^\ast}$, the desired inequality holds
trivially. Otherwise, $d$ ranks $\pi_{\hat k}$ above $\pi_{k^\ast}$, so the pair $\{\hat k,k^\ast\}$ is misordered and
contributes $U_{k^\ast}-U_{\hat k}$ to $\Loss(d,U)$. Since all terms of \eqref{eq:loss} are nonnegative,
\[
  0\le U_{k^\ast}-U_{\hat k}\le\Loss(d,U).
\]
This inequality holds pointwise for every realized ranking, including any external randomization. Taking expectations
yields the second claim.
\hfill$\square$
\par\medskip\textbf{Proof of Theorem~\ref{thm:finite}.}
By Lemma~\ref{lem:indep}, given $\Sb=\mathbf s$ the $X_{k,r}$ are independent, take values in $[0,1]$, and have
distributions
$Q_k$ free of $\mathbf s$ with means $U_k$. For $U_k>U_{k'}$ and whatever the tie-breaking order, $\dES$ ranks
$\pi_{k'}$ above
$\pi_k$ only if $D_{kk'}\ge0$, where $D_{kk'}:=\Uh_{k'}-\Uh_k$. Given $\Sb=\mathbf s$, $D_{kk'}$ is a sum of $C_{k'}$
independent terms with ranges of length $1/C_{k'}$ and $C_k$ with ranges of length $1/C_k$, and
$\E[D_{kk'}\mid\Sb=\mathbf s]=-\Delta_{kk'}$, so the inequality of \citet{hoeffding1963probability} gives~(i)
conditionally on
$\Sb$, with a bound free of $\mathbf s$, hence unconditionally. Weighting~(i) by $\Delta_{kk'}$ and summing over pairs
gives~(ii). The last claim follows from
Proposition~\ref{prop:dominate}. \hfill$\square$

\medskip\chg\textbf{Proof of Theorem~\ref{thm:rate}.}
Write $\beta_{kk'}=\alpha_k\alpha_{k'}/(\alpha_k+\alpha_{k'})$, so that $C_kC_{k'}/(C_k+C_{k'})=C\beta_{kk'}$.
(i) With $h_k=C^{1/2}U_k$ each term of~\eqref{eq:finite-bound} equals
\rev{$C^{-1/2}|h_k-h_{k'}|\exp\{-2(h_k-h_{k'})^2\beta_{kk'}\}$}.
Summing and taking the supremum over $h\in\R^K$ gives $\kappa_K(\alpha)C^{-1/2}$. Since $D_{kk'}\ge-1$,
$\1\{D_{kk'}\ge0\}\le1+D_{kk'}$, so a pair contributes at most $\Delta_{kk'}(1-\Delta_{kk'})\le1/4$, and the sum is
at most $T/4$. As \rev{$\max_{t\ge0}t\exp(-2\beta t^2)=(4e\beta)^{-1/2}$}, bounding each term of~\eqref{eq:kappa} by its
own maximum gives $\kappa_K\le\bar\kappa_K<\infty$.
(ii) $\beta_{kk'}$ is concave in $\alpha$, and \rev{$x\mapsto t\exp(-2t^2x)$ for each $t\ge0$} and $x\mapsto x^{-1/2}$
are
convex and
nonincreasing, so every term of~\eqref{eq:kappa} and of $\bar\kappa_K$ is convex in $\alpha$, and so are the sums and
the supremum. Both functions are invariant to permuting the arms, so Jensen's inequality over the $K!$ permutations
gives a minimum at $\bar\alpha=(1/K,\dots,1/K)$, where $\beta_{kk'}=1/(2K)$ and $\bar\kappa_K=T\{K/(2e)\}^{1/2}$.
(iii) Let $j$ minimize $\alpha_k$, fix $h>0$ and a size $n$ with every $n\pi_k$ an integer, and put
$p_C=1/2+hC^{-1/2}$. Under $P_C$ every cluster has size $n$ and draws $B_i\sim\mathrm{Bernoulli}(p_C)$, and every
unit's outcome is $B_i$ when exactly $n\pi_j$ units are treated and $1/2$ otherwise. Then $U_j=p_C$, $U_l=1/2$ and
$\Uh_l=1/2$ for $l\ne j$, and on $\{\Uh_j<1/2\}$ all $K-1$ pairs $(j,l)$, each with gap $hC^{-1/2}$, are misordered
whatever the tie-breaking order. By Lemma~\ref{lem:indep}, $C_j\Uh_j\sim\mathrm{Binomial}(C_j,p_C)$, so by the
central limit theorem $\P(\Uh_j<1/2)\to\Phi(-2h\alpha_j^{1/2})$ and
$\liminf_CC^{1/2}\sup_{P\in\mathcal P}\Reg(\dES,P)\ge(K-1)h\Phi(-2h\alpha_j^{1/2})$ for every $h>0$. Maximizing over
$h$, with
$u=2h\alpha_j^{1/2}$, gives $(K-1)\astar/(2\alpha_j^{1/2})$. \hfill$\square$
% Printed proofs for Section 4 (label app:s4 is set in main.tex).
\par\medskip\noindent\chg Throughout the proofs of Section~\ref{sec:optimality}, $\Ybh$ denotes the cluster mean of a
generic cluster observation with size $n$ and assignment $Z$, $f(u)=u\Phi(-u)$, and for an observation in arm $r$\rev{,}
$\phi_r=m(n)-U+e_r\{\Ybh-\mu_r(n,Z)\}/\alpha_r\in\R^K$, with $m(n)=(m_1(n),\dots,m_K(n))'$ and all moments taken
under $P^0$ unless stated otherwise. By the argument of Lemma~\ref{lem:indep}, conditionally on $\Sb$ the observations
in the arm positions are independent with distributions $P_r$, so all statements below are conditional on $\Sb$ and do
not
depend on its value. A rule may depend on $\Sb$. Because the position-indexed distributions do not, the Bayes risk
conditional
on $\Sb=\mathbf s$ is bounded below uniformly in $\mathbf s$, so it suffices to treat rules of the position-indexed
observations.

\medskip\textbf{Proof of Lemma~\ref{lem:efficient}.}
\emph{Expansion.} The supported groups $(k,n,z)$ are finitely many and each has probability
$p_n\gamma(\pi_k,z\mid n)>0$ among arm-$k$ observations, so a group is empty with probability that is exponentially
small in $C$. Write $\widehat C_{knz}$ for the group count. Then
$\widehat C_{knz}/C_k=p_n\gamma(\pi_k,z\mid n)+O_P(C^{-1/2})$ and $\widehat\mu_k(n,z)-\mu_k(n,z)=O_P(C^{-1/2})$.
Expanding the products in \eqref{eq:Ueff},
\[
 \Ueff_k-U_k=\sum_n(\widehat p_n-p_n)m_k(n)+\frac1{C_k}\sum_{i:S_i=k}\{\Ybh_i-\mu_k(n_i,\Zb_i)\}+O_P(C^{-1}),
\]
so $\sqrt C(\Ueff-U)=C^{-1/2}\sum_{i}\phi_{S_i}(O_i)+o_P(1)$, with $C_k/C$ in place of $\alpha_k$ in $\phi$. The
summands are independent, bounded and centred. Since $\E\{\Ybh-\mu_r(n,Z)\mid n,Z\}=0$,
$\E_r\phi_r\phi_r'=\operatorname{Cov}_p\{m(n)\}+e_re_r'w_r/\alpha_r^2$, and
$\sum_r\alpha_r\E_r\phi_r\phi_r'=V$. The Lindeberg central limit theorem gives the limit $\Normal(0,V)$.

\emph{Efficiency and raw means.} The submodel of Step~1 below is LAN with information $V^{-1}$ and has welfare
derivative $I_K$, so by the convolution theorem (\citealp{van2000asymptotic}) every regular estimator of $U$ has
asymptotic covariance at least $V$. $\Ueff$ is asymptotically linear with influence function $\phi_r=V(A+B_r)$, which is
$V$ times the score of that submodel, so it is regular and attains the bound, and $V$ is the efficiency bound. The raw
mean's influence function is $\psi_r=e_r(\Ybh-U_r)/\alpha_r$, and
$\sum_r\alpha_r\E_r\psi_r\psi_r'=\operatorname{diag}(v_k/\alpha_k) =V_{\mathrm{raw}}$. Using $\E(\Ybh\mid
n,Z)=\mu_r(n,Z)$ and $\sum_z\gamma(\pi_r,z\mid n)\mu_r(n,z)=m_r(n)$, a direct computation gives
$\sum_r\alpha_r\E_r\psi_r\phi_r'=V$, so $D_r=\psi_r-\phi_r$ satisfies $\sum_r\alpha_r\E_rD_r\phi_r'=0$ and
$V_{\mathrm{raw}}=V+\sum_r\alpha_r\E_rD_rD_r'\succeq V$. In the reduced experiment the arm distributions are
unrestricted and the same argument with one group per arm gives the bound $V_{\mathrm{raw}}$, so the raw means are
efficient there. If $\mu_r(n,z)=U_r$ on every supported group, then $D_r=0$ for every $r$ and $V=V_{\mathrm{raw}}$.
\hfill$\square$

\medskip\textbf{Proof of Theorem~\ref{thm:nearminimax}(i).}
\emph{Step 1: a least favourable submodel.} For the size and the arm-$r$ outcome, let
$A(n)=V^{-1}\{m(n)-U\}$ and
$B_r=V^{-1}e_r\{\Ybh-\mu_r(n,Z)\}/\alpha_r$. For $t\in\R^K$ let $P_t$ tilt the common size distribution to
$p_n(t)\propto p_n\exp\{t'A(n)\}$ and each conditional outcome distribution in arm $r$ to
$F_{rnz,t}(dy)\propto\exp(t'B_r)F_{rnz}(dy)$. The assignment probabilities are unchanged, so
$P_t\in\mathcal P_{\mathrm{loc}}$. Put $P_{C,h}=P_{h/\sqrt C}$. The arm-$r$ score at $t=0$ is $A+B_r=V^{-1}\phi_r$.
Both tilts use bounded statistics, so the log normalizers are $\tfrac12t'\operatorname{Var}(A)t+O(|t|^3)$ and
$\tfrac12t'\operatorname{Var}(B_r\mid n,z)t+O(|t|^3)$ uniformly. Summing over clusters,
\[
 \log\frac{dP_{C,h}}{dP^0}=h'\Delta_C-\tfrac12h'V^{-1}h+o_P(1),\qquad
 \Delta_C=C^{-1/2}\sum_iV^{-1}\phi_{S_i}(O_i)\to\Normal(0,V^{-1}),
\]
because $\operatorname{Var}(A)+\sum_r\alpha_r\E\operatorname{Var}(B_r\mid n,Z)=V^{-1}\{\operatorname{Cov}_p(m)
+\operatorname{diag}(w_k/\alpha_k)\}V^{-1}=V^{-1}$. Hence, for every finite $H\subset\R^K$, the experiments
$(P_{C,h}:h\in H)$ converge to $(\Normal(h,V):h\in H)$ (\citealp{van2000asymptotic}). The welfare derivative is
\[
\partial U(P_t)/\partial t'|_{t=0}=\operatorname{Cov}_p\{m(n),A(n)\}+\operatorname{diag}(w_k/\alpha_k)V^{-1}=I_K,
\]
so $\sqrt C\{U(P_{C,h})-u_0\1\}=h+O(C^{-1/2})$. Finally
$C\sum_k\alpha_kH^2(P_{C,h,k},P_k^0)\to\tfrac14h'V^{-1}h$ (Lemma~\ref{lem:S1-hel}).
Hence $P_{C,h}\in\Hnb_C(b)$ for all large $C$ once $4b^2>\max_{h\in H}h'V^{-1}h$.

\chg \emph{Step 2: finite priors.} The loss \eqref{eq:loss} depends on \rev{$U$} only through differences and is
positively
homogeneous, so $\sqrt C\,\Loss(d,U(P_{C,h}))=\Loss(d,\sqrt C\{U(P_{C,h})-u_0\1\})\to\Loss(d,h)$ for each ranking $d$.
Fix a prior $\pi$ on a finite $H$. For each $C$ choose $\delta_C$ whose Bayes risk
$\sum_h\pi(h)\sqrt C\,\Reg(\delta_C,P_{C,h})$ is within $C^{-1}$ of the infimum over rules. Along a subsequence
attaining the $\liminf$ of these Bayes risks, pass to a further subsequence on which
$P_{C,h}(\delta_C=d)\to q_h(d)$ for every $h\in H$ and every ranking $d$. By the asymptotic representation theorem
(\citealp{van2000asymptotic}), there is a randomized rule $\delta$ in the limit problem with
$P_h(\delta=d)=q_h(d)$. Hence the Bayes risks converge to $\sum_h\pi(h)\Reg_V(\delta,h)\ge r_V(\pi)$, the Bayes
risk of $\pi$ in the limit problem. As the maximum over $H$ exceeds the $\pi$-average, the left side of~(i)
is at least $\sup\{r_V(\pi):\pi\text{ has finite support}\}$.

\emph{Step 3: the Gaussian prior on the welfare vector.} Let $\Pi_\tau=\Normal(0,\tau^2V)$. Then
$h\mid X\sim\Normal(\kappa X,\kappa V)$ with $\kappa=\tau^2/(1+\tau^2)$. For $d\in\{0,1\}^T$, the posterior
expected loss is $\sum_{k<k'}[\E\{(h_k-h_{k'})_+\mid X\}-d^{kk'}\E(h_k-h_{k'}\mid X)]$, which is minimized by
$d^{kk'}=\1\{\kappa(X_k-X_{k'})>0\}$, the ES comparison. ES is a ranking, so it is Bayes among all randomized
rankings. Under $\Pi_\tau$, $h_k-h_{k'}=\tau s_{kk'}Z$ with $Z\sim\Normal(0,1)$, so
\[
 r_V(\Pi_\tau)=\sum_{k<k'}\E|h_k-h_{k'}|\Phi\!\Big(-\frac{|h_k-h_{k'}|}{s_{kk'}}\Big)
 =\sum_{k<k'}2\tau s_{kk'}\!\int_0^\infty\! z\varphi(z)\Phi(-\tau z)\,dz,
\]
where $\varphi$ is the standard normal density, which equals $Q(V)(2\pi)^{-1/2}\{\tau-\tau^2(1+\tau^2)^{-1/2}\}$ after
integration by parts. The term in braces has zero derivative where
$(1+\tau^2)^{3/2}=\tau(2+\tau^2)$. Squaring gives $\tau^4+\tau^2-1=0$, so $\tau^2=(\sqrt5-1)/2$ and
$r_V(\Pi_\tau)=\bstar Q(V)$.

\emph{Step 4: discretization.} The risks $\Reg_V(\delta,\cdot)$ are equicontinuous over $\delta$ on compact sets,
because $|P_h(\delta=d)-P_{h'}(\delta=d)|$ is at most the total variation distance between $\Normal(h,V)$ and
$\Normal(h',V)$ and $h\mapsto\Loss(d,h)$ is Lipschitz, and they are bounded by the $\Pi_\tau$-integrable function
$\sum_{k<k'}|h_k-h_{k'}|$. Truncating $\Pi_\tau$ to a compact set, partitioning it into finitely many pieces on which
every risk varies by less than $\varepsilon/2$, and placing the mass of each piece on one of its points gives a
finite-support prior $\pi_\varepsilon$ with $r_V(\pi_\varepsilon)\ge r_V(\Pi_\tau)-\varepsilon$
(Lemma~\ref{lem:S1-prior} gives the details). Hence $\sup_\pi r_V(\pi)\ge\bstar Q(V)$, and
Step~2 gives the display in~(i). The last claim follows from Step~1. \hfill$\square$

\medskip\textbf{Proof of Theorem~\ref{thm:nearminimax}(ii)--(iii).}
(ii) Fix $b$. For each $C$ choose $P_C\in\Hnb_C(b)$ whose scaled risk $\sqrt C\,\Reg(\deff,P_C)$ is within
$C^{-1}$ of the supremum over $\Hnb_C(b)$. Since total variation is bounded by
the Hellinger distance and $\Ybh\in[0,1]$, $h_C=\sqrt C\{U(P_C)-u_0\1\}$ is bounded, and the group probabilities,
means and variances under $P_C$ converge to those under $P^0$. The expansion in the proof of Lemma~\ref{lem:efficient}
therefore holds under $P_C$, and along any subsequence with $h_C\to h$ each pair with $h_k\ne h_{k'}$ is misordered
with probability tending to $\Phi(-|h_k-h_{k'}|/s_{kk'})$, while pairs with $h_k=h_{k'}$ have vanishing weight
(Lemmas~\ref{lem:S1-tv}--\ref{lem:S1-risk} give the uniformity details). Hence
$\sqrt C\,\Reg(\deff,P_C)\to\sum_{k<k'}|h_k-h_{k'}|\Phi(-|h_k-h_{k'}|/s_{kk'})\le\MES(V)$ along a further subsequence
of every subsequence. The second inequality holds because each term equals
$s_{kk'}f(|h_k-h_{k'}|/s_{kk'})\le\astar s_{kk'}$. The same argument applies to raw means with $V_{\mathrm{raw}}$.

(iii) By (i), $\Reg^{\mathrm{loc}}_*\ge\bstar Q(V)$. Since the infimum defining $\Reg^{\mathrm{loc}}_*$ runs over
sequences that include $\deff$, (ii) gives
$\Reg^{\mathrm{loc}}_*\le\lim_{b}\limsup_C\sup_{P\in\Hnb_C(b)}\sqrt C\,\Reg(\deff,P)\le\astar Q(V)\le
(\astar/\bstar)\Reg^{\mathrm{loc}}_*$, with $\astar/\bstar=1.418843$. \hfill$\square$

\medskip\chg\textbf{Proof of Corollary~\ref{cor:twosat}.} Let
$h_0=(u^*/s_{12})V(e_1-e_2)$. Under the two-point prior on $\pm h_0$, the posterior favours $+h_0$ if and only if
$X'V^{-1}h_0>0$, that is\rev{,} $X_1>X_2$, so ES
is Bayes, with Bayes risk $u^*s_{12}\Phi(-u^*)=\astar s_{12}$. Step~2 gives $\Reg^{\mathrm{loc}}_*\ge\astar s_{12}$
and part~(ii) gives equality. \hfill$\square$

\medskip\chg\textbf{Proof of Theorem~\ref{thm:allocation}.}
(i) Each term of $G$ is $\|(\sqrt{v_k/m_k},\sqrt{v_{k'}/m_{k'}})\|$, a Euclidean norm, which is convex and strictly
increasing in each positive coordinate, composed with the strictly convex maps $m_k\mapsto\sqrt{v_k/m_k}$. Along a
segment $m+t\,d$ with $d_k\ne0$, every term containing $k$ is therefore strictly convex in $t$. Every $k$ belongs to
some pair, so $G$ is strictly convex on $(0,\infty)^K$. As $m_k\downarrow0$ with $v_k>0$, $G\to\infty$. The budget set
has compact closure, so a minimizer exists, lies in the interior, and is unique. The problem is convex with a linear
constraint, so the Lagrange conditions are necessary and sufficient. Since
$\partial G/\partial m_k=-\tfrac12(v_k/m_k^2)\sum_{k'\ne k}(v_k/m_k+v_{k'}/m_{k'})^{-1/2}$, they are
\eqref{eq:alloc-foc} with $\lambda$ equal to twice the multiplier. Lower bounds $m_k\ge\underline m$ add linear
constraints, and the Karush--Kuhn--Tucker conditions apply.
(ii) If $v_k$ and $c_k$ do not depend on $k$, every permutation of $m^*$ is also a minimizer. By uniqueness $m^*$ is
invariant under permutations, so it is balanced.
For $K=2$ the sum in \eqref{eq:alloc-foc} has the single term $(v_1/m_1+v_2/m_2)^{-1/2}$, common to $k=1,2$.
Equating the two conditions gives $v_1/(c_1m_1^2)=v_2/(c_2m_2^2)$, that is\rev{,} $m_1/m_2=\{v_1c_2/(v_2c_1)\}^{1/2}$.
(iii) Each term of $G$ is increasing in $v_k$, and $v_k\le1/4$ for means in $[0,1]$, so the supremum is at $v_k=1/4$
for every $m$. With equal costs, the first claim of (ii) applies. \hfill$\square$

% S1: uniformity details for the proof of Theorem 2 (thm:nearminimax), Appendix A (app:s4).
% Macros \MES, \Ueff, \deff, \Hnb are in macros.tex.

\section{Regularity for the local asymptotic results}\label{sec:S1}

The proof of Theorem~\ref{thm:nearminimax} in Appendix~\ref{app:s4} uses three facts that are routine but
need care because they must hold uniformly over the Hellinger neighbourhoods $\Hnb_C(b)$: the expansion of $\Ueff$ under
a sequence of distributions in $\Hnb_C(b)$, the convergence of misordering probabilities along such a sequence, and the
Hellinger expansion of the tilted submodel. This section supplies them and gives the finite-support approximation of
Step~4. Notation is that of Section~\ref{sec:optimality}: $\Ybh$ is the
observed mean of a generic cluster, $\mathcal P$ is the model of Assumptions~\ref{ass:model}, \ref{ass:sampling}
and~\ref{ass:local}, and all statements are conditional on $\Sb$, given which the observations in the arm positions
are independent with distributions that do not depend on $\Sb$ (Lemma~\ref{lem:indep}).

\subsection{Uniform expansion of the efficient estimator}\label{sec:S1-unif}

\begin{lemma}\label{lem:S1-tv}
\chg Fix $b>0$ and let $P_C\in\Hnb_C(b)$ for every $C$. Then:
(i) $\sup_A|P_{C,k}(A)-P^0_k(A)|\le b(\alpha_kC)^{-1/2}$ for every arm $k$.
(ii) $|U_k(P_C)-u_0|\le b(\alpha_kC)^{-1/2}$, so $h_C=\sqrt C\{U(P_C)-u_0\1\}$ is bounded.
(iii) \rev{The} size probabilities, group probabilities $p_n\gamma(\pi_k,z\mid n)$, group means $\mu_k(n,z)$, group
variances,
and hence $w_k$, $v_k$ and the covariance $V(P_C)$ of \eqref{eq:effcov}, converge to their values under $P^0$.
\end{lemma}

\begin{proof}
The total variation distance is bounded by the Hellinger distance,
\[
\sup_A|F(A)-G(A)|\le H(F,G),\qquad H^2(F,G)=\int(\sqrt{dF}-\sqrt{dG})^2,
\]
and $\alpha_kH^2(P_{C,k},P_k^0)\le b^2/C$ gives (i). Since $\Ybh\in[0,1]$, (ii) follows from (i). The size and group
probabilities are probabilities of fixed events, so they converge by (i). Group
probabilities are bounded away from zero under $P^0$ and there are finitely many groups, so conditional means and
variances of the bounded variable $\Ybh$ within a group are ratios of converging quantities with denominators
bounded away from zero.
\end{proof}

\begin{lemma}\label{lem:S1-exp}
Under the conditions of Lemma~\ref{lem:S1-tv},
\[
\sqrt C\{\Ueff-U(P_C)\}=C^{-1/2}\sum_{i=1}^C\phi_{S_i}^{(C)}(O_i)+o_{P_C}(1)\rightsquigarrow\Normal(0,V),
\]
where $\phi^{(C)}_r=m(n)-U+e_r\{\Ybh-\mu_r(n,Z)\}/\alpha_r$ is evaluated at $P_C$ and $V$ is evaluated at $P^0$.
\end{lemma}

\begin{proof}
By Lemma~\ref{lem:S1-tv}(iii) there is $\varepsilon>0$ such that every supported group $(k,n,z)$ has probability at
least $\varepsilon$ under $P_{C,k}$ for all large $C$. The count of a group among the $C_k$ arm-$k$ observations is
binomial, so by Hoeffding's inequality it is zero with probability at most $\exp(-2C_k\varepsilon^2)$, uniformly in
$C$, and the
probability that some group is empty tends to zero. On the complement, write $\widehat\mu_k(n,z)=\bar a/\bar c$ with
$\bar a$ the arm-$k$ average of $\Ybh\,\1\{(n_i,\Zb_i)=(n,z)\}$ and $\bar c$ the arm-$k$ frequency of the group. Both
averages have bounded summands, so their deviations from their $P_C$-means are $O_{P_C}(C^{-1/2})$ with constants
that depend only on the bounds. A second-order expansion of $\bar a/\bar c$ around the means, whose denominator exceeds
$\varepsilon$, leaves a remainder $O_{P_C}(C^{-1})$. Likewise $\widehat p_n-p_n=O_{P_C}(C^{-1/2})$. The expansion in
the proof of Lemma~\ref{lem:efficient} therefore holds under $P_C$ with an $O_{P_C}(C^{-1})$ remainder, and replacing
$C_k/C$ by $\alpha_k$ in $\phi$ costs $o_{P_C}(1)$ after scaling. The leading term is a sum of independent, bounded,
centred vectors. Its covariance $\sum_r(C_r/C)\E_{P_C}\phi^{(C)}_r\phi^{(C)\prime}_r$ tends to $V$ by
Lemma~\ref{lem:S1-tv}(iii). The Lindeberg condition holds because the summands are bounded and $C\to\infty$, so the
Lindeberg--Feller theorem gives the limit. The constants in these $O_{P_C}$ bounds depend only on the bound
$\Ybh\in[0,1]$, the floor $\varepsilon$, the number of groups and $\min_kC_k/C$, none of which varies with $C$ or
with the sequence $P_C$. The radius $b$ enters only through the index after which the floor holds.
\end{proof}

\begin{lemma}\label{lem:S1-risk}
Under the conditions of Lemma~\ref{lem:S1-tv}, along any subsequence with $h_C\to h$,
\[
\sqrt C\,\Reg(\deff,P_C)\to\sum_{k<k'}|h_k-h_{k'}|\,\Phi\bigl(-|h_k-h_{k'}|/s_{kk'}(V)\bigr).
\]
The same holds for raw ES with $V_{\mathrm{raw}}$ in place of $V$.
\end{lemma}

\begin{proof}
The scaled loss of pair $(k,k')$ is $|h_{C,k}-h_{C,k'}|$ times the indicator that $\deff$ misorders the pair. If
$h_k=h_{k'}$, the weight tends to zero and the indicator is bounded. If $h_k>h_{k'}$, then for large $C$ the pair is
misordered when $T_C=\sqrt C\{(\Ueff_k-\Ueff_{k'})-(U_k-U_{k'})(P_C)\}\le-(h_{C,k}-h_{C,k'})$, up to the tie convention.
By Lemma~\ref{lem:S1-exp}, $T_C\rightsquigarrow\Normal(0,s_{kk'}^2)$ with $s_{kk'}>0$. The limit distribution function
is continuous, so convergence is uniform in its argument (P\'olya's theorem). With $h_{C,k}-h_{C,k'}\to h_k-h_{k'}$
the probability tends to $\Phi(-(h_k-h_{k'})/s_{kk'})$, and the probability of an exact tie tends to zero. The raw
means have the expansion of Lemma~\ref{lem:S1-exp} with influence function $e_r(\Ybh-U_r)/\alpha_r$, by the same
argument with a single group per arm.
\end{proof}

Lemma~\ref{lem:S1-risk} is the step used in part~(ii) of the proof in Appendix~\ref{app:s4}: choosing $P_C$ within
$C^{-1}$ of the
supremum over $\Hnb_C(b)$ and applying the lemma along every subsequence gives the $\limsup$ bound $\MES(V)$.

\subsection{The tilted submodel}\label{sec:S1-tilt}

Step 1 of the proof in Appendix~\ref{app:s4} tilts the size distribution by $A(n)=V^{-1}\{m(n)-U\}$ and each
conditional outcome
distribution in arm $r$
by $B_r=V^{-1}e_r\{\Ybh-\mu_r(n,z)\}/\alpha_r$. Both are bounded, and $\E_pA=0$, $\E(B_r\mid n,z)=0$.

\begin{lemma}\label{lem:S1-hel}
As $|t|\to0$, the tilted distributions $P_t$ satisfy
\[
\textstyle\sum_k\alpha_kH^2(P_{t,k},P_k^0)=\tfrac14t'V^{-1}t+O(|t|^3).
\]
Hence
$C\sum_k\alpha_kH^2(P_{C,h,k},P^0_k)\to\tfrac14h'V^{-1}h$.
\end{lemma}

\begin{proof}
For a bounded statistic $s$ with mean zero and the tilt $dF_t=\exp\{t's-\psi(t)\}dF$, where
$\psi(t)=\log\E_F\exp(t's)$,
\[
H^2(F_t,F)=2-2\E_F\exp\{\tfrac12t's-\tfrac12\psi(t)\}=2-2\exp\{\psi(t/2)-\tfrac12\psi(t)\}.
\]
Since $\psi(t)=\tfrac12t'\E_F(ss')t+O(|t|^3)$, the exponent is $-\tfrac18t'\E_F(ss')t+O(|t|^3)$, so
$H^2(F_t,F)=\tfrac14t'\E_F(ss')t+O(|t|^3)$. The arm-$r$ distribution of an observation is the product of the tilted size
distribution
and the tilted conditional outcome distribution (the assignment distribution is not tilted), so the affinity is
$\sum_n\sqrt{p_n(t)p_n}\sum_z\gamma(\pi_r,z\mid n)\{1-\tfrac12H^2(F_{rnz,t},F_{rnz})\}$. Writing
$\sqrt{p_n(t)p_n}=p_n+O(|t|)$ and expanding both Hellinger distances shows that to second order the two
contributions add:
\[
H^2(P_{t,r},P^0_r)=\tfrac14t'\{\operatorname{Var}(A)+\E\operatorname{Var}(B_r\mid n,Z)\}t+O(|t|^3).
\]
Weighting by $\alpha_r$ and summing gives
$\tfrac14t'[\operatorname{Var}(A)+\sum_r\alpha_r\E\operatorname{Var}(B_r\mid n,Z)]t=\tfrac14t'V^{-1}t$, by the
information identity in Step 1 of the proof in Appendix~\ref{app:s4}.
\end{proof}

\subsection{Finite-support priors}\label{sec:S1-steps}

Step~4 of the proof of Theorem~\ref{thm:nearminimax}(i) approximates the Gaussian prior by priors with finite
support. Lemma~\ref{lem:S1-prior} gives this approximation in detail.

\begin{lemma}\label{lem:S1-prior}
For every $\varepsilon>0$ there is a prior $\pi$ with finite support such that, for every randomized ranking rule
$\delta$,
\[
\sum_h\pi(h)\Reg_V(\delta,h)\ge\int\Reg_V(\delta,h)\,d\Pi_\tau(h)-\varepsilon.
\]
Hence $r_V(\pi)\ge r_V(\Pi_\tau)-\varepsilon$.
\end{lemma}

\begin{proof}
Put $\bar L(h)=\sum_{k<k'}|h_k-h_{k'}|\ge\Loss(d,h)$ for every ranking $d$. $\bar L$ is $\Pi_\tau$-integrable.
Choose $R$ with $\int_{|h|>R}\bar L\,d\Pi_\tau\le\varepsilon/3$. For $|h|,|h'|\le R$ and every $\delta$,
\begin{align*}
|\Reg_V(\delta,h)-\Reg_V(\delta,h')|
&\le\sum_{k<k'}\bigl|(h_k-h_{k'})-(h'_k-h'_{k'})\bigr|+\bar L(h')\,\|\Normal(h,V)-\Normal(h',V)\|_{\mathrm{TV}}\\
&\le c_R|h-h'|,
\end{align*}
where $\|\cdot\|_{\mathrm{TV}}$ is the $L_1$ distance between the densities, with $c_R$ independent of $\delta$: the
first term bounds the change of the loss at a fixed ranking, because the loss of a fixed pair decision is
$(h_k-h_{k'})_+$ or $(h_{k'}-h_k)_+$, which is 1-Lipschitz in the signed contrast. The second bounds the change of the
distribution of $X$ for a loss in $[0,\bar L(h')]$.
Partition $\{|h|\le R\}$ into finitely many sets $A_j$ of diameter at most $\varepsilon/(3c_R)$, pick
$h_j\in A_j$ and put $\pi(h_j)=\Pi_\tau(A_j)$. Then for every $\delta$,
$\sum_j\pi(h_j)\Reg_V(\delta,h_j)\ge\int_{|h|\le R}\Reg_V(\delta,h)\,d\Pi_\tau-\varepsilon/3
\ge\int\Reg_V(\delta,h)\,d\Pi_\tau-2\varepsilon/3$. Normalizing $\pi$ to total mass one divides by
$\Pi_\tau(|h|\le R)\le1$ and does not decrease the left side. Taking the infimum over $\delta$ gives the last claim.
\end{proof}
% S2: computing kappa_K (eq:kappa of the main text) and the integer-allocation statement of Section 4.
% thm:rate and thm:allocation are proved in the main text (appA_s3, appA_s4).

\section{The constant $\kappa_K$ and integer allocations}\label{sec:S2}

This section computes the constant $\kappa_K$ of Theorem~\ref{thm:rate} and proves the integer-allocation
statement of Section~\ref{sec:optimality}.

\subsection{Computing $\kappa_K$}\label{sec:S2-kappa}

\chg Theorem~\ref{thm:rate} bounds the maximum risk of ES by $\min\{\kappa_K(\alpha)C^{-1/2},T/4\}$, with $\kappa_K$
defined in \eqref{eq:kappa} and the pairwise bound $\bar\kappa_K(\alpha)=(2\sqrt e)^{-1}\sum_{k<k'}
\{(\alpha_k+\alpha_{k'})/(\alpha_k\alpha_{k'})\}^{1/2}$, which maximizes each pair's term over its own gap. The
pairwise bound has a closed form, while $\kappa_K$ requires a maximization. Write
$\beta_{kk'}=\alpha_k\alpha_{k'}/(\alpha_k+\alpha_{k'})$ and
\rev{$F(h)=\sum_{k<k'}|h_k-h_{k'}|\exp\{-2(h_k-h_{k'})^2\beta_{kk'}\}$}.

\begin{lemma}\label{lem:S2-kappa}
(i) The supremum of $F$ over $\R^K$ is attained. (ii) For $K\ge3$, $\kappa_K(\alpha)<\bar\kappa_K(\alpha)$.
\end{lemma}

\begin{proof}
\chg (i) $F$ is continuous and invariant to common shifts of $h$. Attainment holds for $K=2$. Suppose it holds for fewer
than $K$ arms and that a maximizing sequence leaves every compact set modulo shifts. Along a subsequence the arms split
into two nonempty groups whose distance diverges. The cross terms vanish, so the limit value is at most
$\kappa_A+\kappa_B$, the constants of the two groups, which are attained by the induction hypothesis. Placing the two
maximizers at a large finite distance adds strictly positive cross terms, so $\sup F>\kappa_A+\kappa_B$, a
contradiction. (ii) Each term of $F$ is at most its maximum $(4e\beta_{kk'})^{-1/2}$, attained only at
\rev{$|h_k-h_{k'}|=(4\beta_{kk'})^{-1/2}$}. Equality at a maximizer $h^*$ would need every pair at its own maximizer.
For three arms ordered by $h^*$ as $k$, $l$, $m$ this gives \rev{$|h^*_k-h^*_m|=|h^*_k-h^*_l|+|h^*_l-h^*_m|$}, but
$(4\beta_{km})^{-1/2}<(4\beta_{kl})^{-1/2}+(4\beta_{lm})^{-1/2}$, since squaring the right side gives
$\tfrac14(1/\alpha_k+2/\alpha_l+1/\alpha_m)+2(4\beta_{kl})^{-1/2}(4\beta_{lm})^{-1/2}>\tfrac14(1/\alpha_k+1/\alpha_m)$.
\end{proof}

\emph{Algorithm.} By shift invariance we fix $h_1=0$, and by relabelling it is enough to search the ordered chamber
with nonnegative adjacent gaps for each of the $K!/2$ orderings up to reversal (one ordering under balance). The
constant $\kappa_K$ is computed by numerical multistart maximization: a local optimizer is run from many random
starting points and the best value is kept. The reported $\kappa_K$ is therefore a numerical value, not a certified
enclosure, whereas the analytical constant $\bar\kappa_K$ is a certified upper bound. Under balance with $K=4$,
$\kappa_4=4.3077$. Table~\ref{tab:S2-kappa} gives the values for $K=2,\dots,6$.

\begin{table}[ht]
\centering
\caption{Uniform constants under balanced allocation}\label{tab:S2-kappa}
\begin{tabular*}{\textwidth}{@{\extracolsep{\fill}}lcccccc}
\toprule
$K$ & $T/4$ & $\bar\kappa_K$ & $\kappa_K$ & $\kappa_K/\bar\kappa_K$ & $(4\bar\kappa_K/T)^2$ & $(4\kappa_K/T)^2$ \\
\midrule
2 & 0.25 & 0.6065 & 0.6065 & 1 & 5.89 & 5.89 \\
3 & 0.75 & 2.2285 & 1.9863 & 0.891 & 8.83 & 7.01 \\
4 & 1.50 & 5.1466 & 4.3077 & 0.837 & 11.77 & 8.25 \\
5 & 2.50 & 9.5901 & 7.7123 & 0.804 & 14.72 & 9.52 \\
6 & 3.75 & 15.7581 & 12.3262 & 0.782 & 17.66 & 10.80 \\
\bottomrule
\end{tabular*}
\begin{minipage}{\textwidth}
\footnotesize
\emph{Notes:} $T=K(K-1)/2$ and $\bar\kappa_K=T\{K/(2e)\}^{1/2}$. The last two columns give the number of clusters
below which the elementary bound $T/4$ is smaller than the corresponding $C^{-1/2}$ bound. For $K=2$ the two
constants coincide because there is one pair. $\kappa_K$ is the multistart maximum described in the text, computed by
two independent implementations (\texttt{codes/checks/kappa\_check.py} and \texttt{codes/checks/kappa\_chamber.py})
that agree to four decimals.
\end{minipage}
\end{table}

\subsection{Integer allocations}\label{sec:S2-alloc}

Theorem~\ref{thm:allocation}(i) characterizes the allocation that minimizes $G(m;v)$ under the budget by the
first-order conditions \eqref{eq:alloc-foc}. Since $G$ is strictly convex on the positive orthant and the budget
set is linear, any standard convex solver (sequential quadratic programming with the budget as an equality
constraint, started at the balanced allocation) returns the unique solution, and convergence can be checked by
verifying that the left sides of \eqref{eq:alloc-foc} agree across $k$. Section~\ref{sec:optimality} states
that, under equal variances and costs, the integer allocations whose counts differ by at most one are the integer
minimizers of $G$.

\begin{proposition}\label{prop:S2-integer}
Let $v_k=v$ and $c_k=c$ for all $k$, and let the total $C=B/c\ge K$ be an integer. The minimizers of $G(\cdot\,;v)$
over integer allocations with $\sum_km_k=C$ are exactly the allocations whose counts differ by at most one.
\end{proposition}

\begin{proof}
If $m_k\ge m_{k'}+2$, then $m-e_k+e_{k'}$ lies strictly inside the segment from $m$ to the allocation with $m_k$ and
$m_{k'}$ swapped, which has the same value of $G$ by symmetry, so the strict convexity established in the printed
proof of Theorem~\ref{thm:allocation}(i) gives $G(m-e_k+e_{k'})<G(m)$. Allocations whose counts differ by at most
one are permutations of each other and have the same value.
\end{proof}
% S4: full simulation designs and tables (supports sec:sim, fig:sim).
% The number of replications is stated once, in Section 5 of the main text, and is not restated here.

\section{Simulation designs and full results}\label{sec:S4}

This section gives the designs behind Figure~\ref{fig:sim} and the tables of numbers that the figure summarizes. All
designs use the replications stated in Section~\ref{sec:sim}, and the code is in the directory
\texttt{codes/sim\_ectj/} of the replication package.

\subsection{Common features}\label{sec:S4-common}

All designs have $K=4$ saturations $\Pib=(0.1,0.3,0.5,0.7)$, binary unit outcomes, and the two-stage design of
Assumption~\ref{ass:sampling}: $C=4m$ clusters with $m\in\{2,4,8,16,32,64\}$ per arm under complete
randomization, then complete randomization within clusters with exactly $n_i\pi_k$ treated units. Cluster sizes are
drawn uniformly from $\{10,20,30,40\}$, so $n_i\pi_k$ is an integer for every saturation
(Assumption~\ref{ass:model}). Welfare, ES decisions and regret are computed under the exact-count assignment
distribution
above. For fixed-gap designs we report the standardized gap $\Delta/\sigma$ with
$\Delta=|U_k-U_{k'}|$ and $\sigma^2=(v_k+v_{k'})/m$, where $v_k$ is the variance of the cluster mean in arm $k$.

\subsection{Beta-Binomial cluster design}\label{sec:S4-bb}

Cluster $i$ in arm $k$ draws $P_i\sim\mathrm{Beta}\{qU_k,q(1-U_k)\}$, and its units have outcomes that are
independent Bernoulli$(P_i)$ given $P_i$. Only the arm-specific distribution of the cluster mean enters the risk
(Lemma~\ref{lem:indep}), so this specifies the design completely. The intracluster correlation is $1/(q+1)$, the
welfare of saturation $k$ is exactly $U_k$, and the variance of the cluster mean is
\[
v_k=U_k(1-U_k)\Bigl\{\frac1{q+1}+\frac{q}{q+1}\E(1/n_i)\Bigr\},
\]
since $\operatorname{Var}(P_i)=U_k(1-U_k)/(q+1)$ and the within-cluster binomial variance contributes
$\E\{P_i(1-P_i)\}\E(1/n_i)$, with $\E\{P_i(1-P_i)\}=U_k(1-U_k)\,q/(q+1)$. With sizes uniform on $\{10,20,30,40\}$,
$\E(1/n_i)=25/480\approx0.0521$. Local designs set $U_k=1/2+h_k/\sqrt C$ for $h$ in a finite grid $H$ of welfare
vectors with $\max_k|h_k|<\sqrt2$, so that every $U_k$ lies in $(0,1)$ for all $C\ge8$. A refined grid $H'$ is used
for the grid-refinement check. Fixed-gap designs set $U$ directly.

\chg We use $q\in\{19,3\}$ \rev{(intracluster correlations 0.05 and 0.25)}. The grid $H$ consists of the $h$ with
entries in
$\{-1.2,-1.0,\ldots,1.2\}$ (step 0.2), and $H'$ of those with entries in $\{-1.4,-1.3,\ldots,1.4\}$ (step 0.1), so
$\max_k|h_k|\le1.4<\sqrt2$. The arms are exchangeable in this design, and with lower-index tie-breaking the increasing
order of $h$ is the worst, so only increasing $h$ are searched. A check over all permutations agrees up to Monte Carlo
error. For the allocation design of Table~\ref{tab:S4-alloc}, where the arms are not exchangeable, $H$ is searched in
every order and the refinement adds the step-$0.1$ neighbours of the 40 best points of $H$. The fixed-gap vectors are
$U=(0.425,0.475,0.525,0.575)$ (gaps 0.05), $U=(0.35,0.45,0.55,0.65)$ (gaps 0.10) and $U=(0.40,0.50,0.52,0.65)$
(uneven gaps).

\subsection{Network design}\label{sec:S4-network}

Inside each cluster an unobserved stochastic block model graph links units. A unit's outcome is a logistic threshold
in its own treatment and the treated share of its graph neighbours. The graph is never used by the
estimator or the ranking rule. Welfare $U(\pi_k)$ is not read off the coefficients: it is the expectation of the
cluster mean under the exact-count assignment distribution, computed by a separate high-precision simulation whose
numerical
error is
reported in the notes to Table~\ref{tab:S4-network}.

Each cluster has two blocks (membership probability $1/2$), with edge probabilities 0.35 within and 0.05 between
blocks. Unit $j$ has outcome $Y_j=\1\{\varepsilon_j\le\beta_0+a_i+\beta_1Z_j+\beta_2E_j+\beta_3Z_jE_j\}$, where
$\varepsilon_j$ is logistic, $E_j$ is the treated share of $j$'s neighbours ($0$ if $j$ has none) and
$a_i\sim\Normal(0,0.5^2)$ is a cluster effect. Network 1 has $\beta=(-0.5,0.5,2.2,-3.0)$, which gives
$U=(0.4406,0.5110,0.5241,0.4819)$, an interior optimum at $\pi=0.5$ and sorted adjacent gaps 0.041, 0.029 and 0.013.
Network 2 has $\beta=(-0.2,0.6,0.4,0)$, which gives $U=(0.4760,0.5219,0.5673,0.6120)$ and gaps 0.046, 0.045 and
0.045. Welfare is computed from $4\times10^7$ clusters per saturation, with numerical standard error at most
$1.8\times10^{-5}$.

\subsection{Monte Carlo error control}\label{sec:S4-mc}

Before running the designs we fix a familywise error level of $5\%$ over the $J$ comparisons (design, gap vector)
between Monte Carlo risk and the bound of Theorem~\ref{thm:finite}. For comparison $j$ with $M_j$ independent
replications and per-replication loss in $[0,D_j]$, $D_j=\sum_t\Delta_{t,j}$, the one-sided lower confidence limit is
\[
L_j=\widehat\Reg_j-D_j\sqrt{\log(J/0.05)/(2M_j)}.
\]
A comparison is flagged if $L_j$ exceeds the bound. There are $J=78$ comparisons (24 local, 36 fixed-gap, 12 network
and 6 allocation designs) and no flag is raised: the largest ratio of Monte Carlo risk to bound is 0.37.

\subsection{Results}\label{sec:S4-results}

Tables~\ref{tab:S4-rate} to~\ref{tab:S4-network} report the full results summarized in Section~\ref{sec:sim}.

\begin{table}[ht]
\centering
\caption{Scaled maximum regret, Beta-Binomial design}\label{tab:S4-rate}
\small\setlength{\tabcolsep}{3pt}
\begin{tabular*}{\textwidth}{@{\extracolsep{\fill}}lcccccc}
\toprule
 & $m=2$ & $m=4$ & $m=8$ & $m=16$ & $m=32$ & $m=64$ \\
\midrule
\multicolumn{7}{l}{Intracluster correlation 0.05 ($q=19$)} \\
$\sqrt C\max_{h\in H}\widehat\Reg$ & 0.404 & 0.400 & 0.399 & 0.398 & 0.398 & 0.398 \\
Monte Carlo s.e. & 0.0004 & 0.0004 & 0.0004 & 0.0004 & 0.0004 & 0.0004 \\
Refined grid $H'$ & 0.405 & 0.399 & 0.398 & 0.398 & 0.397 & 0.399 \\
Monte Carlo s.e. & 0.0004 & 0.0004 & 0.0004 & 0.0004 & 0.0004 & 0.0004 \\
Gaussian limit on $H$ / $H'$ & \multicolumn{6}{c}{0.398 / 0.398} \\
$\astar Q$ / $\bstar Q$ & \multicolumn{6}{c}{0.455 / 0.321} \\
\midrule
\multicolumn{7}{l}{Intracluster correlation 0.25 ($q=3$)} \\
$\sqrt C\max_{h\in H}\widehat\Reg$ & 0.665 & 0.667 & 0.668 & 0.670 & 0.670 & 0.671 \\
Monte Carlo s.e. & 0.0007 & 0.0007 & 0.0007 & 0.0007 & 0.0007 & 0.0007 \\
Refined grid $H'$ & 0.675 & 0.671 & 0.675 & 0.675 & 0.675 & 0.675 \\
Monte Carlo s.e. & 0.0006 & 0.0007 & 0.0006 & 0.0006 & 0.0007 & 0.0006 \\
Gaussian limit on $H$ / $H'$ & \multicolumn{6}{c}{0.671 / 0.676} \\
$\astar Q$ / $\bstar Q$ & \multicolumn{6}{c}{0.775 / 0.547} \\
\bottomrule
\end{tabular*}
\begin{minipage}{\textwidth}
\footnotesize
\emph{Notes:} $U_k=1/2+h_k/\sqrt C$, balanced design, $C=4m$. $H$ is the set of $h$ with entries on the grid
$\{-1.2,-1.0,\ldots,1.2\}$ and $H'$ the refinement with step $0.1$ on $[-1.4,1.4]$. Each entry is the Monte Carlo
regret at the maximizing $h$, selected in a separate search and evaluated on independent replications.
$Q=\sum_{k<k'}(v_k/\alpha_k+v_{k'}/\alpha_{k'})^{1/2}$ with $\alpha_k=1/4$ and $v_k$ at $U_k=1/2$. The Gaussian limit is
the
maximum over the grid of $\sum_{k<k'}d_{kk'}\Phi(-d_{kk'}/\sigma)$, $\sigma^2=8v$. For comparison,
$\MES(V)=0.399$ and 0.679. Scaled by $\sqrt C$, the bound of Theorem~\ref{thm:rate} is $\min\{\kappa_4,\sqrt C\,T/4\}$
with
$\kappa_4=4.3077$.
\end{minipage}
\end{table}

\begin{table}[ht]
\centering
\caption{Theorem~\ref{thm:finite} bound and Monte Carlo risk at fixed gaps}\label{tab:S4-bound}
\small\setlength{\tabcolsep}{3pt}
\begin{tabular*}{\textwidth}{@{\extracolsep{\fill}}lcccccc}
\toprule
 & $m=2$ & $m=4$ & $m=8$ & $m=16$ & $m=32$ & $m=64$ \\
\midrule
\multicolumn{7}{l}{$U=(0.425,0.475,0.525,0.575)$} \\
Bound of Theorem~\ref{thm:finite} & 0.4887 & 0.4778 & 0.4569 & 0.4192 & 0.3567 & 0.2688 \\
MC risk, $q=19$ & 0.1376 & 0.0998 & 0.0640 & 0.0354 & 0.0164 & 0.0054 \\
\quad s.e. & 0.00012 & $9.9{\times}10^{-5}$ & $7.3{\times}10^{-5}$ &
  $4.7{\times}10^{-5}$ & $2.8{\times}10^{-5}$ & $1.6{\times}10^{-5}$ \\
\quad ratio bound/risk & 3.6 & 4.8 & 7.1 & 11.8 & 21.8 & 49.7 \\
\quad $\max_{k<k'}\Delta_{kk'}/\sigma_{kk'}$ & 0.96 & 1.36 & 1.92 & 2.72 & 3.85 & 5.44 \\
\quad smallest adjacent $\Delta/\sigma$ & 0.32 & 0.45 & 0.63 & 0.90 & 1.27 & 1.80 \\
MC risk, $q=3$ & 0.1800 & 0.1518 & 0.1187 & 0.0824 & 0.0496 & 0.0254 \\
\quad s.e. & 0.00013 & 0.00013 & 0.00011 & $8.8{\times}10^{-5}$ & $6.0{\times}10^{-5}$ & $3.7{\times}10^{-5}$ \\
\quad ratio bound/risk & 2.7 & 3.1 & 3.8 & 5.1 & 7.2 & 10.6 \\
\quad $\max_{k<k'}\Delta_{kk'}/\sigma_{kk'}$ & 0.56 & 0.80 & 1.13 & 1.60 & 2.26 & 3.19 \\
\quad smallest adjacent $\Delta/\sigma$ & 0.19 & 0.26 & 0.37 & 0.53 & 0.74 & 1.05 \\
\midrule
\multicolumn{7}{l}{$U=(0.35,0.45,0.55,0.65)$} \\
Bound of Theorem~\ref{thm:finite} & 0.9139 & 0.8384 & 0.7134 & 0.5376 & 0.3459 & 0.1901 \\
MC risk, $q=19$ & 0.1294 & 0.0698 & 0.0319 & 0.0104 & 0.0015 & $4.0{\times}10^{-5}$ \\
\quad s.e. & 0.00015 & $9.3{\times}10^{-5}$ & $5.6{\times}10^{-5}$ &
  $3.1{\times}10^{-5}$ & $1.2{\times}10^{-5}$ & $2.0{\times}10^{-6}$ \\
\quad ratio bound/risk & 7.1 & 12.0 & 22.4 & 51.7 & 230 & 4,740 \\
\quad $\max_{k<k'}\Delta_{kk'}/\sigma_{kk'}$ & 1.99 & 2.82 & 3.99 & 5.64 & 7.98 & 11.28 \\
\quad smallest adjacent $\Delta/\sigma$ & 0.64 & 0.90 & 1.27 & 1.80 & 2.55 & 3.61 \\
MC risk, $q=3$ & 0.2377 & 0.1633 & 0.0977 & 0.0497 & 0.0202 & 0.0049 \\
\quad s.e. & 0.00022 & 0.00017 & 0.00012 & $7.4{\times}10^{-5}$ & $4.3{\times}10^{-5}$ & $2.2{\times}10^{-5}$ \\
\quad ratio bound/risk & 3.8 & 5.1 & 7.3 & 10.8 & 17.1 & 38.8 \\
\quad $\max_{k<k'}\Delta_{kk'}/\sigma_{kk'}$ & 1.17 & 1.65 & 2.34 & 3.31 & 4.68 & 6.62 \\
\quad smallest adjacent $\Delta/\sigma$ & 0.37 & 0.53 & 0.75 & 1.06 & 1.50 & 2.11 \\
\midrule
\multicolumn{7}{l}{$U=(0.40,0.50,0.52,0.65)$} \\
Bound of Theorem~\ref{thm:finite} & 0.7243 & 0.6826 & 0.6097 & 0.4962 & 0.3506 & 0.2042 \\
MC risk, $q=19$ & 0.1301 & 0.0754 & 0.0361 & 0.0142 & 0.0069 & 0.0048 \\
\quad s.e. & 0.00014 & $9.9{\times}10^{-5}$ & $6.2{\times}10^{-5}$ &
  $3.1{\times}10^{-5}$ & $1.3{\times}10^{-5}$ & $8.6{\times}10^{-6}$ \\
\quad ratio bound/risk & 5.6 & 9.1 & 16.9 & 34.9 & 51.2 & 42.9 \\
\quad $\max_{k<k'}\Delta_{kk'}/\sigma_{kk'}$ & 1.64 & 2.32 & 3.28 & 4.64 & 6.56 & 9.27 \\
\quad smallest adjacent $\Delta/\sigma$ & 0.13 & 0.18 & 0.25 & 0.36 & 0.51 & 0.72 \\
MC risk, $q=3$ & 0.2134 & 0.1575 & 0.1017 & 0.0549 & 0.0237 & 0.0095 \\
\quad s.e. & 0.00019 & 0.00016 & 0.00012 & $8.1{\times}10^{-5}$ & $4.6{\times}10^{-5}$ & $2.0{\times}10^{-5}$ \\
\quad ratio bound/risk & 3.4 & 4.3 & 6.0 & 9.0 & 14.8 & 21.4 \\
\quad $\max_{k<k'}\Delta_{kk'}/\sigma_{kk'}$ & 0.96 & 1.36 & 1.92 & 2.72 & 3.85 & 5.44 \\
\quad smallest adjacent $\Delta/\sigma$ & 0.07 & 0.11 & 0.15 & 0.21 & 0.30 & 0.42 \\
\bottomrule
\end{tabular*}
\begin{minipage}{\textwidth}
\footnotesize
\emph{Notes:} Beta-Binomial design, balanced, $C=4m$, $T=6$. The bound does not depend on $q$.
$\sigma_{kk'}^2=(v_k+v_{k'})/m$. Ratio is the bound divided by the Monte Carlo risk. The cap $T/4=1.5$ of
Theorem~\ref{thm:rate} never binds. Panel (b) of Figure~\ref{fig:sim} plots the rows for $q=3$.
\end{minipage}
\end{table}

\begin{table}[ht]
\centering
\caption{Allocation comparison at a fixed total number of clusters}\label{tab:S4-alloc}
\small\setlength{\tabcolsep}{3pt}
\begin{tabular*}{\textwidth}{@{\extracolsep{\fill}}lcccc}
\toprule
Allocation & Counts $(C_1,\ldots,C_4)$ & $\astar G$ & Max.\ MC regret over $H$ (s.e.) & Refined (s.e.) \\
\midrule
Balanced & (15, 15, 15, 15) & 0.0942 & 0.0868 (0.0001) & 0.0872 (0.0001) \\
Minimizer of $G$ & (10, 12, 17, 21) & 0.0909 & 0.0814 (0.0001) & 0.0821 (0.0001) \\
Tilted & (21, 17, 13, 9) & 0.1062 & 0.1000 (0.0001) & 0.1001 (0.0001) \\
\bottomrule
\end{tabular*}
\begin{minipage}{\textwidth}
\footnotesize
\chg\emph{Notes:} Beta-Binomial design with $C=60$ clusters, equal costs and $q=(19,9,3,1)$, so the intracluster
correlations are $(0.05,0.10,0.25,0.50)$ and $v=(0.0249,0.0367,0.0723,0.1315)$ at $U_k=1/2$.
$G(m)=\sum_{k<k'}(v_k/m_k+v_{k'}/m_{k'})^{1/2}$. The minimizer of $G$ is $(10.07,12.10,16.46,21.37)$, rounded by largest
remainders. The tilted allocation is proportional to $v_k^{-1/2}$, \rev{toward} the low-variance arms.
$U_k=1/2+h_k/\sqrt
C$. $H$ has entries on $\{-1.2,-1.0,\ldots,1.2\}$ in every order, and the refinement adds the step-$0.1$ neighbours
of the 40 best points of $H$. Relative to the balanced design the maximum regret over $H$ changes by $-6.2$\%
(minimizer) and $+15.3$\% (tilted).
\end{minipage}
\end{table}

\begin{table}[ht]
\centering
\caption{Network design}\label{tab:S4-network}
\small\setlength{\tabcolsep}{3pt}
\begin{tabular*}{\textwidth}{@{\extracolsep{\fill}}lcccccc}
\toprule
 & $m=2$ & $m=4$ & $m=8$ & $m=16$ & $m=32$ & $m=64$ \\
\midrule
\multicolumn{7}{l}{Network 1: $U=(0.4406,0.5110,0.5241,0.4819)$} \\
Theorem~\ref{thm:finite} bound & 0.2774 & 0.2753 & 0.2710 & 0.2627 & 0.2472 & 0.2199 \\
MC risk & 0.0999 & 0.0843 & 0.0655 & 0.0452 & 0.0270 & 0.0140 \\
\quad s.e. & $7.8{\times}10^{-5}$ & $7.3{\times}10^{-5}$ & $6.4{\times}10^{-5}$ &
  $5.1{\times}10^{-5}$ & $3.4{\times}10^{-5}$ & $2.1{\times}10^{-5}$ \\
Normal approximation & 0.0992 & 0.0841 & 0.0655 & 0.0452 & 0.0269 & 0.0140 \\
Smallest adjacent $\Delta/\sigma$ & 0.08 & 0.12 & 0.17 & 0.24 & 0.33 & 0.47 \\
\midrule
\multicolumn{7}{l}{Network 2: $U=(0.4760,0.5219,0.5673,0.6120)$} \\
Theorem~\ref{thm:finite} bound & 0.4450 & 0.4368 & 0.4210 & 0.3918 & 0.3420 & 0.2680 \\
MC risk & 0.1338 & 0.1008 & 0.0678 & 0.0393 & 0.0192 & 0.0073 \\
\quad s.e. & 0.00011 & $9.7{\times}10^{-5}$ & $7.4{\times}10^{-5}$ &
  $5.0{\times}10^{-5}$ & $3.0{\times}10^{-5}$ & $1.8{\times}10^{-5}$ \\
Normal approximation & 0.1305 & 0.0996 & 0.0671 & 0.0390 & 0.0192 & 0.0073 \\
Smallest adjacent $\Delta/\sigma$ & 0.29 & 0.40 & 0.57 & 0.81 & 1.14 & 1.62 \\
\bottomrule
\end{tabular*}
\begin{minipage}{\textwidth}
\footnotesize
\emph{Notes:} Welfare computed under the exact-count assignment distribution from 40,000,000 clusters per saturation.
Its
numerical
standard error is at most $1.8{\times}10^{-5}$. Each cluster has two blocks (membership probability $1/2$), edge
probabilities $0.35$ within and $0.05$ between blocks, and outcome
$Y_j=\1\{\varepsilon_j\le\beta_0+a_i+\beta_1Z_j+\beta_2E_j+\beta_3Z_jE_j\}$ with logistic $\varepsilon_j$, $E_j$ the
treated share of $j$'s neighbours ($0$ if none) and $a_i\sim\Normal(0,0.5^2)$. Network 1: $\beta=(-0.5,0.5,2.2,-3.0)$,
welfare peaks at $\pi=0.5$. Network 2: $\beta=(-0.2,0.6,0.4,0)$, welfare increasing. Normal approximation:
$\sum_{k<k'}\Delta_{kk'}\Phi(-\Delta_{kk'}/\sigma_{kk'})$ with simulated $v_k$.
\end{minipage}
\end{table}
% S5: data construction (Zomba) and sensitivity checks for Section 6 of the main text.

\section{Data construction and additional application results}\label{sec:S5}

\subsection{The Zomba analysis sample}\label{sec:S5-zomba}

\chg \emph{Source.} We use the public-use file of the Zomba \rev{cash-transfer} experiment (Harvard Dataverse,
\url{https://doi.org/10.7910/DVN/7K1A25}, file \texttt{corc.dta}), the data of \citet{baird2011cash}. The randomized
saturation component of the design is also discussed by \citet{baird2018optimal} and \citet{baird2014designing}. Code
is in \texttt{codes/empirical\_zomba/} of the replication package.

\emph{Design.} The study sampled enumeration areas (EAs) in Zomba district, Malawi, and randomized them within three
geographic strata: Zomba city (29 EAs), rural EAs within 16 km of the city (119) and rural EAs beyond 16 km (28)
(\citealp{baird2011cash}). Of the 176 EAs, 88 are programme EAs and 88 are pure-control EAs. Within programme EAs,
baseline
schoolgirls sampled into the study were offered transfers at a nominal saturation of $0$, $1/3$, $2/3$ or $1$,
separately for conditional and unconditional cash transfers (CCTs and UCTs, respectively). The offer count was ``the
integer number of treatments that made the EA-level sample saturation as close as possible to that assigned''
(\citealp{baird2014designing}), and offers were randomized within the EA. All programme EAs also offered CCTs to
baseline dropouts.

\chg \rev{\emph{Transfers.} The sampled girls were never married, aged 13 to 22 and enrolled in school at baseline. An
offer
specified monthly transfers to the girl and to her parents or guardians for the 2008 school year and was later extended
to the end of 2009. The parents' amount was randomized across EAs and the girl's amount by a lottery within the EA.
Under a CCT, the secondary school fees of recipients were also paid directly to the school. A UCT offer paid no fees,
but households with girls eligible for secondary school received an addition to the transfer equal to the average fees
paid in the CCT arm (\citealp{baird2011cash}).}

\chg \emph{Sample.} In the public file the variable \texttt{group} identifies CCT offer recipients (\texttt{T2a}), UCT
offer
recipients (\texttt{T2b}), schoolgirls without an offer in programme EAs (\texttt{S2}) and schoolgirls in pure-control
EAs (\texttt{C2}). For each programme EA we compute the offer share among \texttt{T2a}, \texttt{T2b} and \texttt{S2}
and map it to the nominal saturations. An EA belongs to the CCT arm if its offers are CCT, to the UCT arm if they are
UCT, and to the zero-offer group if it has none. One EA (EA 218: near-rural, 10 sampled girls, 8 UCT and 2 CCT
offers) is counted as UCT saturation $1$, as in \citet{baird2014designing}, whose cell totals (9 UCT EAs and
128 UCT recipients at saturation $1$, and 46 CCT EAs, 15/16/15 at $1/3$, $2/3$, $1$) are reproduced exactly with this
classification. Dropping it gives UCT counts 9/9/8 and reverses only the near-tied top pair of the UCT never-pregnant
ranking. The public file is a balanced panel in which every girl carries the same \texttt{group} code in rounds 1--3,
so offer shares computed from it \rev{use} the baseline denominators. In 85 of the 88 programme EAs the offer count
equals
the nearest integer to $\pi_kn_i$ for exactly one nominal saturation (for $\pi_k\in\{1/3,2/3\}$ it is never a
half-integer). EA 194 ($n_i=1$, no offer)
fits both $0$ and $1/3$ and is a zero-offer EA, and EA 157 (CCT, 9 of 25 offers) and EA 237 (UCT, 20 of 31) differ
by one and are kept at the saturations $1/3$ and $2/3$ that \citet{baird2014designing} assign them. Dropping them
changes no ES
ranking. No EA is excluded. The original assignment lists are not public, so the cause of the two off-by-one counts
cannot be verified.

\emph{Outcomes and attrition.} Three binary outcomes are used: enrolment in school in the first term of 2010
(\texttt{inschool\_\allowbreak term1\_\allowbreak 2010}), never married
($1-{}$\texttt{ever\_\allowbreak married}) and never pregnant ($1-{}$\texttt{ever\_\allowbreak pregnant}).
Each is averaged over the observed schoolgirls of an EA to give the cluster mean. Of the 2{,}907 baseline schoolgirls,
127 (4.4\%) were not surveyed at round 3, and one more lacks enrolment and marriage. The missing share is 2.3 to
4.7\% across the seven programme groups (5.0\% for enrolment in CCT $1/3$, against 4.8\% in pure control), and we do
not reject equal attrition across saturations (EA-level Kruskal--Wallis $p=0.94$ in the CCT menu, $0.87$ in the UCT
menu). Cluster means are complete-case means. No EA loses all its girls. Three programme EAs have fewer than three
observed girls (EA 194 in the zero arm with $n=1$, EA 124 in CCT $1/3$ with 2 of 4, and EA 191 in CCT $1$ with 2 of
2). They are kept, since each cluster mean is still a bounded observation.

\emph{Sampling weights.} Baseline schoolgirls were sampled at age- and stratum-specific rates (14--45\% in urban and
70--100\% in rural EAs, \citealp{baird2011cash}). The file's \texttt{wgt} ranges from 1 to 26.5 and varies
within 146 of the 176 EAs. Our welfare is the unweighted mean over sampled girls, that is, welfare for the sampled
population. The correlation between the EA offer share and the EA mean weight is $-0.02$. With weighted EA means only
the never-married rankings change (CCT $1/3,0,2/3,1$ and UCT $1,1/3,2/3,0$). Restricting to girls with
\texttt{panel}${}=1$ changes only UCT never married ($1,1/3,2/3,0$) and CCT enrolment ($2/3,1/3,0,1$).

\emph{Policies and the zero arm.} The 15 programme EAs with no schoolgirl offers form the observed $0\%$ arm for CCTs:
they are programme EAs whose dropouts received CCT offers, which is what a $0\%$ schoolgirl CCT saturation means. Using
the same 15 EAs as the $0\%$ arm of the UCT ranking requires the policy-equivalence assumption that a programme with
$0\%$ UCT offers to schoolgirls, together with CCT offers to dropouts, is the same policy as the observed zero-offer
programme. The 88 pure-control EAs represent a different policy (no programme) and are not a saturation level of
either menu.

\chg \emph{Welfare under the implemented assignment distribution.} Nominal shares times EA sample sizes are generally
not
integers, so the integrality condition of Assumption~\ref{ass:model} is replaced by the assignment rule actually
used, which offers transfers to the integer number of sampled schoolgirls closest to $\pi$ times their number, chosen
by complete randomization. Welfare $U(\pi)$ is defined under this assignment distribution, so the arm means of cluster
means remain unbiased for it. The first stage was blocked on the three strata. Conditional on the stratum labels and
the assignment, the cluster means within each arm remain independent and bounded, so Theorem~\ref{thm:finite}
applies to the arm targets $\sum_s(C_{k,s}/C_k)U_k^{(s)}$ of Section~\ref{sec:application}.
Table~\ref{tab:S5-counts} gives the stratum composition. The CCT cells are close to proportional. The UCT cells are
not: far rural has 0, 1 and 0 UCT EAs at the positive saturations against 3 in the shared zero arm, so a
three-stratum weighted UCT mean is not estimable, and we keep the pooled arm means. Under city/rural
\rev{poststratification} (programme-EA shares $15/88$ and $73/88$) all UCT rankings are unchanged and each CCT ranking
swaps
one adjacent near-tied pair.

\begin{table}[ht]
\centering
\caption{Zomba EA counts by arm, saturation and stratum}\label{tab:S5-counts}
\begin{tabular*}{\textwidth}{@{\extracolsep{\fill}}llcccc}
\toprule
Arm & Saturation & EAs & City & Near rural & Far rural \\
\midrule
Zero offer (both menus) & $0$ & 15 & 3 & 9 & 3 \\
CCT & $1/3$ & 15 & 2 & 10 & 3 \\
CCT & $2/3$ & 16 & 3 & 9 & 4 \\
CCT & $1$ & 15 & 2 & 10 & 3 \\
UCT & $1/3$ & 9 & 2 & 7 & 0 \\
UCT & $2/3$ & 9 & 1 & 7 & 1 \\
UCT & $1$ & 9 & 2 & 7 & 0 \\
Pure control (no programme) & -- & 88 & 14 & 60 & 14 \\
\bottomrule
\end{tabular*}
\begin{minipage}{\textwidth}
\footnotesize
\emph{Notes:} The zero-offer EAs are shared by both menus (see text), so $C=61$ for CCTs and $C=42$ for UCTs. Counts
reproduce the design of \citet{baird2014designing}. EA 218 (8 UCT and 2 CCT offers) is counted
as UCT saturation $1$. Dropping it gives UCT counts 9/9/8. No EA is excluded. Near rural: within 16 km of Zomba city.
\end{minipage}
\end{table}

\begin{table}[ht]
\centering
\caption{Pairwise standardized gaps, Zomba}\label{tab:S5-gaps}
\begin{tabular*}{\textwidth}{@{\extracolsep{\fill}}llcccccc}
\toprule
Transfer & Outcome & $(0,\tfrac13)$ & $(0,\tfrac23)$ & $(0,1)$ & $(\tfrac13,\tfrac23)$ & $(\tfrac13,1)$ & $(\tfrac23,1)$
\\
\midrule
CCT & Enrolled & 0.30 & $-0.58$ & 0.10 & $-1.01$ & $-0.26$ & 0.79 \\
CCT & Never married & $-0.85$ & $-0.17$ & 0.21 & 0.63 & 1.17 & 0.38 \\
CCT & Never pregnant & 0.65 & 0.50 & 0.81 & $-0.21$ & $-0.09$ & 0.18 \\
UCT & Enrolled & $-0.16$ & 0.33 & $-0.84$ & 0.55 & $-0.84$ & $-1.21$ \\
UCT & Never married & $-0.86$ & $-0.88$ & $-2.14$ & $-0.16$ & $-1.38$ & $-0.97$ \\
UCT & Never pregnant & $-0.16$ & $-0.88$ & $-1.02$ & $-0.75$ & $-0.86$ & 0.11 \\
\bottomrule
\end{tabular*}
\begin{minipage}{\textwidth}
\footnotesize
\emph{Notes:} Entries are $z_{kk'}=(\Uh_k-\Uh_{k'})/\widehat{\mathrm{se}}_{kk'}$ for the pair $(k,k')$ of saturations,
with
$\widehat{\mathrm{se}}_{kk'}=(s_k^2/C_k+s_{k'}^2/C_{k'})^{1/2}$ and $s_k^2$ the sample variance of EA means within
saturation $k$ (the standard errors of Table~\ref{tab:zomba}). The sum of the six $\widehat{\mathrm{se}}_{kk'}$ in a
row is the plug-in objective $\widehat G(C_1,\dots,C_K)$ of Section~\ref{sec:application}.
\end{minipage}
\end{table}

\subsection{Additional regret calculations}\label{sec:S5-regret}

Section~\ref{sec:application} reports the plug-in value $\astar\widehat G$ and the uniform bound of
Theorem~\ref{thm:rate}. The finite-sample bounds also give sufficient numbers of clusters. One pair with gap 0.05 is
ordered correctly with probability 0.9 once there are 922 clusters per saturation, and, by the union bound over the six
pairs, every pair with gap at least 0.05 once there are 1,638 per saturation. A ranking regret of 0.05 over all welfare
vectors needs about 7,400 clusters in total under a balanced design with the numerical multistart value of $\kappa_4$
(Section~\ref{sec:S2-kappa}), or 10,596 with the analytical constant $\bar\kappa_4$, which is a certified upper bound.

\subsection{Sensitivity of the design counterfactual}\label{sec:S5-sens}

The counterfactual allocation in Section~\ref{sec:application} plugs estimated cluster-mean variances $\hat v_k$,
which rest on 9 to 16 clusters per saturation, into the objective $\widehat G$ of Theorem~\ref{thm:allocation}. As
a normal-reference sensitivity check, we replace each $\hat v_k$ in turn by an endpoint of the 95\% chi-square
interval that a normal sample would give for its variance. The plug-in gain of the joint reallocation then ranges from
1.1\% to 8.4\% (0.5--6.7\% for CCTs and 0--9.9\% for UCTs when each menu is reallocated
separately). A within-arm bootstrap of the EAs ($B=1000$, re-optimized) gives
5--95\% ranges of 1.5\% to 8.4\%, and the point-estimate optimum improves on the realized design in more than 95\% of
draws.
% S6: the strict upper bound (a* Q is not the minimax value for K >= 3) and the exchangeable case (Eaton).

\section{The strict upper bound and the exchangeable case}\label{sec:S6}

\chg The proof of Theorem~\ref{thm:nearminimax} (Step 3 and part (ii)) shows $\bstar
Q(V)\le\Reg^*(V)\le\MES(V)\le\astar
Q(V)$ for the Gaussian experiment $X\sim\Normal(h,V)$, where $Q(V)=\sum_{k<k'}s_{kk'}$ and
$s_{kk'}^2=(e_k-e_{k'})'V(e_k-e_{k'})$. The upper end $\astar Q(V)$ is the per-pair value of
\citet{hirano2009asymptotics}
summed over pairs. This section proves that the upper end is strict once $K\ge3$, as stated in
Section~\ref{sec:optimality} of the main text, and proves the exchangeable case stated there. Throughout,
$f(u)=u\Phi(-u)$, $u^*=0.7518$ is its maximizer, $\astar=f(u^*)=0.169971$ and $\bstar=0.119796$. We write $\delta_0$
for the ES ranking in the Gaussian experiment, which orders the arms by the coordinates of $X$ (zero thresholds, ties
broken by the lower index), so that \rev{$\Reg_V(\delta_0,h)=\sum_{k<k'}|h_k-h_{k'}|\Phi(-|h_k-h_{k'}|/s_{kk'})$} and
$\MES(V)=\sup_h\Reg_V(\delta_0,h)$.

\begin{lemma}\label{lem:S6-strict}
Let $K\ge3$ and $V$ be positive definite. (i) The supremum defining $\MES(V)$ is attained. (ii) $\MES(V)<\astar Q(V)$.
\end{lemma}

\begin{proof}
\chg (i) The argument of Lemma~\ref{lem:S2-kappa}(i) applies verbatim, with the pair terms
\rev{$s_{kk'}f(|h_k-h_{k'}|/s_{kk'})$}: each is continuous, positive at positive gaps and tends to zero as its gap
diverges. (ii) Let $h^*$ be a maximizer. Each term is at most $\astar s_{kk'}$, with equality only at
\rev{$|h^*_k-h^*_{k'}|=u^*s_{kk'}$}, $u^*$ the unique maximizer of $f$. If all $T$ pairs attained equality, every gap
would be positive, and for three arms ordered by $h^*$ as $k,l,m$,
\rev{$u^*s_{km}=|h^*_k-h^*_m|=|h^*_k-h^*_l|+|h^*_l-h^*_m|=u^*(s_{kl}+s_{lm})$}. But
$s_{km}=\|V^{1/2}(e_k-e_m)\|\le\|V^{1/2}(e_k-e_l)\|+\|V^{1/2}(e_l-e_m)\|$, with equality only if
$V^{1/2}(e_k-e_l)$ and $V^{1/2}(e_l-e_m)$ are positively proportional, which is impossible because $e_k-e_l$ and
$e_l-e_m$ are linearly independent and $V^{1/2}$ is invertible. Hence $\MES(V)=\Reg_V(\delta_0,h^*)<\astar Q(V)$.
\end{proof}

\begin{proposition}[Exchangeable covariance]\label{prop:S6-exch}
Let $V=aI+b\1\1'$ with $a>0$ and $a+Kb>0$, for any $K$. In the Gaussian limit problem ES is minimax with value
$\Reg^*(V)=\MES(V)=\sqrt a\,\MES(I_K)$, and under the conditions of Theorem~\ref{thm:nearminimax} the ranking
$\deff$ is locally asymptotically minimax with value $\MES(V)$.
\end{proposition}

\begin{proof}
The Gaussian density has property M (\citealp{eaton1967ranking}): the proposition requires
$V^{-1}=c_1I-c_2\1\1'$ with $c_1>0$ and $c_2/c_1<1/K$, and here $c_1=1/a$ and $c_2/c_1=b/(a+Kb)<1/K$. The loss
\eqref{eq:loss} satisfies Eaton's conditions (3.4) and (3.5): swapping two adjacently ranked arms changes only
that pair's term, and the loss is invariant to relabelling arms and actions together. Hence
ES is minimax (\citealp{eaton1967ranking}) and $\Reg^*(V)=\MES(V)$. Every contrast has variance $2a$, so
$\MES(V)=\sqrt a\,\MES(I_K)$. For the transfer to the cluster experiment, let $h^*$ maximize the ES risk, which exists
by Lemma~\ref{lem:S6-strict}(i), and let $\pi^*$ be the uniform prior on its
coordinate permutations. By
\citet{eaton1967ranking} ES is Bayes for every prior whose distribution on the welfare vector is permutation
invariant, so ES is Bayes for $\pi^*$, and because the ES risk is permutation invariant when $V$ is exchangeable the
Bayes risk of $\pi^*$ is $\Reg_V(\delta_0,h^*)=\MES(V)$. Since $\MES(V)$ is the minimax value, $\pi^*$ is least
favourable. It has finite support, so Step~2 of the proof in Appendix~\ref{app:s4} of
Theorem~\ref{thm:nearminimax}(i) gives
$\Reg^{\mathrm{loc}}_*\ge\MES(V)$, and part~(ii) of the theorem gives equality, attained by $\deff$.
\end{proof}

\end{document}